\documentclass[a4paper,11pt,runningheads]{llncs}

\usepackage{fullpage}
\usepackage{amsmath}
\usepackage{amssymb}

\usepackage{multirow}
\usepackage{graphicx}
\usepackage{multirow}
\usepackage{hyperref}
\usepackage{color,xcolor}

\usepackage{tikz}
\usetikzlibrary{calc, matrix, patterns, backgrounds, arrows, shapes, 
chains, fit, decorations, intersections, positioning}

\tikzset{
    shapenode/.style = {draw, rectangle, fill=none, minimum size=0.6cm, minimum height=0.7cm, minimum width=0.7cm, 
    auto, node distance=5em, font=\normalsize}, 
    textnode/.style  = {draw=none, fill=none, rectangle, minimum size=0.6cm, auto, node distance=4.5em, font=\normalsize},
    connect/.style = {black,->}
}

\begin{document}

\title{Threshold Trapdoor Functions and Their Applications}

\author{Binbin Tu\inst{1,2,3,4} \and
	Yu Chen\inst{2,3,4,5} \and
	Xueli Wang\inst{2,3,4}}

\institute{Westone Cryptologic Research Center, Westone Information Industry Inc, Beijing 100070, China \and
	State Key Laboratory of Cryptology, P.O. Box 5159, Beijing 100878, China \and
	State Key Laboratory of Information Security, Institute of Information Engineering, Chinese Academy of Sciences, Beijing, China \and 
	School of Cyber Security, University of Chinese Academy of Sciences, Beijing, China \and
	Ant Financial, Beijing, China\\
	\email{tubinbin,chenyu,wangxueli@iie.ac.cn}}

\maketitle

\begin{abstract}
We introduce a new cryptographic primitive named threshold trapdoor functions (TTDFs), 
from which we give generic constructions of threshold and revocation encryptions under 
adaptive corruption model.  
Then, we show TTDFs can be instantiated under  
the decisional Diffie-Hellman (DDH) assumption and the learning with errors (LWE) assumption.
By combining the instantiations of TTDFs with the generic constructions,
we obtain threshold and revocation encryptions
which compare favorably over existing schemes.
The experimental results show that our proposed schemes are practical.
\end{abstract}

\section{Introduction}\label{sec1}

\begin{trivlist}
\item\textbf{Threshold public-key encryption.} TPKE ~\cite{Des87,DDFY94,SG98,CG99}
can distribute the decryption power among many servers in order to ensure threshold 
servers can decrypt ciphertexts,
while any probabilistic polynomial-time (PPT) adversary
corrupting less than threshold servers is unable to obtain the message.
TPKE itself provides useful functionalities, and 
it is also a significant building block in other cryptographic primitives, 
such as mix-net (anonymous channel)~\cite{Cha81}, public key encryption with
non-interactive opening~\cite{DHKT08,GLFFLMS10}.Generally speaking, a $(n,t)$-TPKE scheme consists of a combiner and $n$ decryption servers. The combiner sends the ciphertext to all servers, any subset of $t$ servers compute the decryption shares and reply, and the combiner combines the replies to obtain the plaintext. However, in the securty model of TPKE, not only the servers may be corrupted, but also the decryption shares could be eavesdropped. Therefore, constructing TPKE schemes by splitting the secret key of public-key encrytion (PKE) directly does not work. Following the generic construction of PKE from trapdoor function (TDF), we try to design a threshold version of TDF for constructing the TPKE scheme, by splitting the master trapdoor into $n$ shares and storing each share on a different server and any subset of $t$ servers can use the shared trapdoors to invert the function, without reconstructing the master trapdoor.

\item\textbf{Revocation public-key encryption.}
RPKE~\cite{NP00,DF03,Wee11} enables a sender to broadcast ciphertexts
and all but some revoked users can
do the decryption. It is a special kind of broadcast encryption~\cite{FN93}
which enables a sender to encrypt messages and transmit ciphertexts to
users on a broadcast channel in order to the chosen users can decrypt ciphertexts.
RPKE has many applications, including
pay-TV systems, streaming audio/video and many others.

Naor and Pinkas~\cite{NP00} considered the revocation scenario: 
a group controller (GC) controls the decryption
capabilities of users. If a subgroup of users is disallowed
to do the decryption, the GC needs to generate a new key which should be known to other users
and be used to encrypt in the further group communication. 
We observe that the threshold version of TDF remains one-way, even if exposing part of shared trapdoors. Therefore, we can revoke this part of shared trapdoors and ensure that any revoked users cannot decrypt the ciphertext.
\end{trivlist}

\subsection{Related Work}
Designing generic construction of TPKE has proved to be a highly 
non-trivial task. Dodis and Katz~\cite{DK05} 
gave a generic construction
of TPKE from multiple encryption technique.
Wee~\cite{Wee11} introduced a new primitive called 
threshold extractable hash proofs and presented 
a generic construction of TPKE from it.
However, both of above constructions are only secure under 
the static corruption model where the adversary can corrupt  
the servers before the scheme is set up.
Following the work of Wee~\cite{Wee11}, Libert and Yung~\cite{LY12} introduced 
a primitive named all-but-one perfectly sound threshold hash 
proof systems, from which they gave a generic 
construction of TPKE under adaptive corruption model 
where the adversary can corrupt the servers at any time.
The results are important since the adaptive 
adversary is strictly stronger than the static one~\cite{CFGN96,Can00}. 
But they only showed concrete instantiations under
number-theoretic assumptions in bilinear groups which are vulnerable to quantum attacks.
Recently, lattices have been recognized as a viable
foundation for quantum-resistant cryptography.
Bendlin and Damg\aa rd~\cite{BD10} gave the first lattice-based TPKE 
based on a variant of Regev's scheme~\cite{Reg05}.
Xie et al.~\cite{XXZ11} designed the first chosen-ciphertext secure (IND-CCA) TPKE 
based on the LWE assumption.
However, both of above TPKEs are only statically secure, 
and the size of the public key and the ciphertext is at least linear in the number of servers.
Bendlin et al.~\cite{BKP13} converted Identity Based Encryption (IBE)~\cite{GPV08} into
threshold one, which can be transformed into a TPKE via the generic transformation in~\cite{BBH06}.
However, in an offline phase, their scheme needs the parties to perform lots of interactive precomputation.
In summary, the state-of-the-art TPKE is not entirely satisfactory.
On one hand, existing generic constructions of TPKE are designed in the limited static
corruption model which fails to capture realistic attacks.
On the other hand, most existing TPKE schemes are based on number-theoretic
assumptions which are insecure against quantum attacks.

As for RPKE, Naor and Pinkas~\cite{NP00} considered a revocation scenario with a group controller and constructed a RPKE scheme under the DDH assumption. 
Unlike the scenario of~\cite{NP00},
Dodis and Fazio~\cite{DF03}
designed a RPKE in which every user who knows
the revoked identities can encrypt messages and every non-revoked user can decrypt ciphertexts.
Then, they constructed IND-CCA RPKE under
the DDH assumption.
Wee~\cite{Wee11} presented a generic construction of RPKE in static corruption model 
and instantiated the construction under the
DDH assumption and factoring assumption respectively.
However,
all of aforementioned schemes are designed under the number-theoretic assumptions
which are insecure against quantum attack.
\subsection{Motivations}
A central goal in cryptography is to construct cryptosystems in strong security models 
which can resist lots of possible attacks.
Another goal is to build cryptosystems under intractability assumptions which
are as general as possible; in this way, 
we can replace the underlying assumption, 
if some assumption is vulnerable to a new attack or if another yields better performance.
Therefore, generic constructions of TPKE and RPKE in stronger adaptive corruption 
model are advantageous.
Meanwhile, with the development of quantum computer,
designing the quantum-resistant TPKE and RPKE is also necessary.
Last but not least, constructing cryptosystems based on the same cryptographic primitive brings additional
advantages such as reducing the footprint of cryptographic code and easily embedding into systems.

Motivated by above discussions, we ask the following challenging questions:\\

\begin{minipage}{0.92\textwidth}
\emph{Can we construct TPKE and RPKE under adaptive corruption
model from one cryptographic primitive?
Can we instantiate this primitive based on quantum-resistant assumptions?}
\end{minipage}

\subsection{Our Contributions}
We introduce
a cryptographic primitive named TTDF, and derive generic constructions
of TPKE and RPKE under adaptive corruption model from it.
Along the way to instantiate TTDF,
we propose a notion called threshold lossy trapdoor function (TLTDF) and prove that TTDF
is implied by TLTDF, while the latter can be instantiated
based on the DDH assumption and the LWE assumption.
Moreover, we show a relaxation of TTDF called threshold trapdoor relation (TTDR), which
enables the same applications of TPKE and RPKE,
and admits more efficient instantiation based on the DDH assumption.
An overview of our contributions is given in Figure~\ref{fig:contributions}.
\begin{figure}[!hbth]
\centering
\begin{tikzpicture}
    \node [textnode, name=TTDF, draw] {TTDF};
    \node [textnode, name=TTDR, right of=TTDF, xshift=6em, draw] {TTDR};

    \node [textnode, name=TPKE, above of=TTDF, xshift=-6em] {TPKE}; 
    \node [textnode, name=RPKE, above of=TTDF, xshift=6em] {RPKE};

    \node [textnode, name=TLTDF, below of=TTDF, draw] {TLTDF};
    \node [textnode, name=TLTDR, right of=TLTDF, xshift=6em, draw] {TLTDR};

    \node [textnode, name=LWE, below of=TLTDF, xshift=-6em] {LWE};
    \node [textnode, name=DDH, below of=TLTDF, xshift=6em] {DDH}; 

    \draw (LWE) edge[->, thick] node[left] {Sec.~\ref{TLTDF Based on LWE}} (TLTDF);
    \draw (DDH) edge[->, thick] node[right] {Sec.~\ref{TLTDF Based on DDH}} (TLTDF); 
    \draw (DDH) edge[->, thick] node[right] {Sec.~\ref{Threshold Trapdoor Relation}} (TLTDR); 

    \draw (TLTDF) edge[->, thick] node[right] {Sec.~\ref{Threshold Lossy Trapdoor Function}} (TTDF);
    \draw (TLTDF) edge[->, dashed] node[above] {Sec.~\ref{Threshold Trapdoor Relation}} (TLTDR);
    \draw (TTDF) edge[->, dashed] node[above] {Sec.~\ref{Threshold Trapdoor Relation}} (TTDR);
    \draw (TLTDR) edge[->, thick] node[right] {Sec.~\ref{Threshold Trapdoor Relation}} (TTDR);

    \draw (TTDF) edge[->, thick] node[left] {Sec.~\ref{Threshold Encryption from TTDF}} (TPKE);
    \draw (TTDF) edge[->, thick] node[right] {Sec.~\ref{Revocation Encryption from TTDF}}(RPKE);     
 
\end{tikzpicture}
\caption{Overview of the results in this work.}
\label{fig:contributions}
\end{figure}
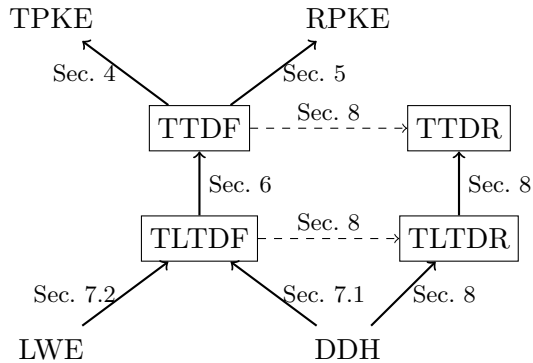

\begin{trivlist}
\item \textbf{Threshold Trapdoor Function.} Informally, TTDF is
a threshold version of trapdoor function. It is parameterized by 
the threshold value $t$ and the number of identities $n$.
$(n,t)$-TTDF splits the master trapdoor into $n$ shared trapdoors.
Every shared trapdoor can be used to compute a piece of inversion share. Then,  
by collecting more than 
$t$ inversion shares, the combiner can recover the preimage. 
Especially, it can even compute inversion shares of any other identity with the help of the preimage.
We formalize security notion for TTDF, namely threshold one-wayness,
which requires that the function remains one-way even
when the adversary can adaptively obtain less than $t$ shared trapdoors.
\end{trivlist}

\begin{trivlist}
\item \textbf{TPKE and RPKE from TTDF.} 
$(n,t)$-TTDF gives rise to a simple construction of $(n,t)$-TPKE and $(n,t-1)$-RPKE.
The main idea of constructing $(n,t)$-TPKE follows constructing public-key encryption from 
trapdoor function. In particular, the sharing algorithm splits the master secret key 
into $n$ shared secret keys, every shared secret key can be used to compute a piece of 
inversion share, and collecting at least $t$ inversion shares can extracts message.
For the security, TTDF holds threshold one-wayness which prevents any PPT adversary who can obtain less than $t$ shared secret keys from decrypting ciphertext, especially under adaptive corruption model. Generally speaking, an adaptive adversary can make the decision of which parties to corrupt at any time during the run of the scheme, in particular, based on the information like the shared trapdoors of corrupted parties gathered. Providing this information is typically the main difficulty in proving adaptive security~\cite{CGJKR99}. TTDF holds a shared trapdoor oracle that given an input of any identity, and outputs a shared trapdoor of this identity. Any PPT adversary can get the information of any corrupted parties by querying the oracle adaptively and obtains at most $t-1$ shared trapdoors.

Following the idea of constructing $(n,t)$-TPKE, we observes that 
$(n,t)$-TPKE holds the security even if exposing $t-1$ shared secret keys. 
Therefore, the GC can revoke any $t-1$ identities by encrypting a new session key and 
exposing their decryption shares.
Every non-revoked user can compute a decryption share and extracts the new session key.
Moreover, threshold one-wayness ensures that 
no PPT adversary can decrypt ciphertext without the non-revoked secret key.
\end{trivlist}

\begin{trivlist}
\item \textbf{Instantiation.} Along the way to instantiate TTDF,
we introduce a new notion called TLTDF, 
which is a threshold version of lossy trapdoor function (LTDF)~\cite{PW08}.
It is parameterized by 
the threshold value $t$ and the number of identities $n$.
Informally,
LTDF has two modes. In the injective mode, it is an injective trapdoor
function. In the lossy mode, it statistically loses an amount of information about the input.
Both of the modes of LTDF are computationally indistinguishable.
However, in both modes of TLTDF,
the master trapdoor can be split into $n$ shared trapdoors and 
every shared trapdoor can be used to compute an inversion share, 
and in the injective mode 
any $t$ inversion shares can be used to retrieve the preimage.
Moreover,
any PPT adversary cannot distinguish both modes,
even when the adversary can adaptively obtain less than $t$ shared trapdoors.

We prove that TTDF is implied by TLTDF and instantiate TLTDF
under the DDH assumption and the LWE assumption respectively.
DDH-based TLTDF is easy to design, while
building LWE-based TLTDF is a non-trivial task.
Intuitively, we transform the inversion algorithm of LTDF into threshold version by
using $(n,t)$-threshold secret sharing scheme~\cite{Sha79}.
Every user gets a shared trapdoor $td_i$, $i\in [n]$, and
computes the inversion share
$\langle a, td_i\rangle+e_i$. Then the combiner obtains $t$ inversion shares to
compute the Lagrangian coefficients $L_i$ for any identity set of size $t$
and recombines the $\langle a, td\rangle$
by computing

\begin{equation*}
\begin{aligned}
L_1(\langle a, td_1\rangle+e_1)+\cdots+L_t(\langle a, td_t\rangle+e_t)&=\langle a, \sum_{i=1}^t L_i\cdot td_i\rangle+\sum_{i=1}^t L_i\cdot e_i\\
&=\langle a, td\rangle+\sum_{i=1}^t L_i\cdot e_i
\end{aligned}
\end{equation*}

Unfortunately, choosing identities in a large identity space causes the norm of errors out of control and
prevents correct inversion.
To resolve this problem, we take advantage of the technique of
``clearing out the denominator''~\cite{Sho00,ABVVW12,BGGK17}.
Note that since the Lagrangian coefficients are
rational numbers and the identity is chosen in $[n]$,
we can scale them to be integers by computing $(n!)^2L_i$.
By instantiating appropriate parameters,
we prove that the quantity of errors preserves bounded, which does not affect the correctness of inversion.
\end{trivlist}

\begin{trivlist}
\item \textbf{Optimization.} We show a relaxation of TTDF called TTDR, 
and prove that TTDR maintains same applications of constructing TPKE and RPKE.
Informally, TTDR replaces the evaluation algorithm of TTDF
with a relation sampling algorithm which can 
generate a random input with its image of a function,
while the function need not be efficiently computable.
We also formalize security notion 
named threshold one-wayness for TTDR following TTDF.

Similarly to instantiating TTDF from TLTDF,
we instantiate TTDR by introducing the notion of threshold lossy trapdoor relation (TLTDR),
which is a threshold version of lossy trapdoor relation (LTDR) \footnote{We give a refined
definition of LTDR in Section \ref{Threshold Trapdoor Relation},
which is more simple and intuitive than the one introduced in~\cite{XLLL14}.}~\cite{XLLL14}.
We prove TTDR is naturally implied by TLTDR.
Moreover, we instantiate TLTDR based on the DDH assumption to obtain an instantiation of TTDR,
which is more efficient than TTDF.
\end{trivlist}

\section{Preliminaries}

\subsection{Notations}
We denote the natural numbers by $\mathbb{N}$, the integers by $\mathbb{Z}$,
the real numbers by $\mathbb{R}$.
We use lower-case bold letters and upper-case bold letters 
to denote vectors and matrices (e.g. $\mathbf{x}$ and $\mathbf{X}$).
Let $\mathbf{x}^T$ and $\mathbf{X}^T$ denote transpositions of 
vector $\mathbf{x}$ and matrix $\mathbf{X}$. 
For $n\in \mathbb{N}$, $1^n$ denotes the string
of $n$ ones, and $[n]$ denotes the set $\{1,\cdots,n\}$. 
We use standard asymptotic $(O,o,\Omega, \omega)$ notation to denote the growth of positive functions.
We denote a negligible function by $\mathsf{negl}(\lambda)$, which is an $f(\lambda)$
such that $f(\lambda)=o(\lambda^{-c})$ for every fixed constant $c$, and we let $\mathsf{poly}(\lambda)$ denote
an unspecified function $f(\lambda)=O(\lambda^c)$ for some constant $c$.
If $S$ is a set then $s\leftarrow S$ denotes the operation of sampling
an element $s$ of $S$ at random.

Let $X$ and $Y$ be two random variables over some countable set $S$.
The statistical distance between X and Y is defined as
$$
\triangle(X,Y)=\frac{1}{2}\sum_{s\in S}|\Pr[X=s]-\Pr[Y=s]|.
$$

\subsection{Assumptions}
\begin{trivlist}
\item \textbf{DDH Assumption.} The generation algorithm 
$\textsf{Gen}$ takes as input a security parameter $1^\lambda$
and outputs $(p, \mathbb{G}, g)$, where $\mathbb{G}$ is a cyclic group of order $p$, 
$p$ is a prime and $g$ is a generator of $\mathbb{G}$.
The DDH assumption \cite{DH76} is that the ensemble $\{(\mathbb{G}, g^a, g^b, g^{ab})\}_{\lambda\in \mathbb{N}}$
and $\{(\mathbb{G}, g^a, g^b, g^{c})\}_{\lambda\in \mathbb{N}}$ 
are computationally indistinguishable, 
where $a, b, c\leftarrow \mathbb{Z}_p$.
\end{trivlist}

\begin{trivlist}
\item \textbf{LWE Assumption.} Let $d$ be the dimension of lattice,
an integer $q=\mathsf{poly}(d)$ and all operations be performed in $\mathbb{Z}_q$.
For an integer dimension $d \in \mathbb{Z}^+$,
a vector $\mathbf{z} \in \mathbb{Z}^d_q$ and 
an error distribution $\chi: \mathbb{Z}_q \rightarrow \mathbb{R}^+$.
$\mathbf{A}_{\mathbf{z},\chi}$ is the distribution 
of the variable $(\mathbf{a}, \langle \mathbf{a}, \mathbf{z}\rangle + e)$
on $\mathbb{Z}^d_q\times \mathbb{Z}_q$,
where $\mathbf{a}\leftarrow \mathbb{Z}^d_q$ and $e\leftarrow \chi$.
The LWE assumption~\cite{Reg05} is that for some secret $\mathbf{z} \in \mathbb{Z}^d_q$ independent samples
from the LWE distribution $\mathbf{A}_{\mathbf{z},\chi}$,
and independent samples from the uniform distribution on $\mathbb{Z}^d_q\times \mathbb{Z}_q$ are
computationally indistinguishable.
\end{trivlist}

\subsection{Randomness Extraction}
We use the notion of average min-entropy~\cite{DORS08}, that 
captures the remaining unpredictability of $X$ conditioned on the value of $Y$:
$$
\widetilde{H}_{\infty}(X|Y)=-\text{lg}(E_{y\leftarrow Y}[2^{-H_{\infty}(X|Y=y)}])
$$
We review the following lemmas from ~\cite{DORS08}.
\begin{lemma}
If $Y$ takes at most $2^r$ values and $Z$ is
any random variable, then $\widetilde{H}_{\infty}(X|(Y,Z))\geq \widetilde{H}_{\infty}(X|Z)-r$.
\end{lemma}
\begin{lemma}
Let $X$, $Y$ be random variables such that $X\in\{0,1\}^l$ and
$\widetilde{H}_{\infty}(X|Y)\geq k$.
Let $H$ be a family of pairwise independent hash functions from $\{0,1\}^l$ to $\{0,1\}^{l'}$.
Then for $h\leftarrow H$, we have
$$
\triangle((Y,h,h(X)),(Y,h,U_{l'}))\leq \epsilon
$$
as long as $l'\leq k-2\lg (1/\epsilon)$.
\end{lemma}

\subsection{Threshold Secret Sharing}
We now recall the threshold secret sharing scheme~\cite{Sha79}. It can be  
parameterized by the number of identities $n$ and the threshold value $t$,
and denotes as $(n,t)$-threshold secret sharing scheme.
Let $\mathbb{F}$ be a finite field, $|\mathbb{F}| > n$.
Let $id_i \in \mathbb{F},i=1,\cdots,n$  be distinct, nonzero elements that are fixed and publicly
known. The scheme works as follows:

\bgroup\catcode`\#=12\gdef\hash{#}\egroup\renewcommand{\labelitemi}{$\bullet$}
\begin{itemize}
  \item $\mathsf{Share}(s,id_i)\rightarrow s_i$:
  On input a secret $s\in \mathbb{F}$, and
  any identity $id_i,i\in[n]$. It chooses $a_1,\cdots,a_{t-1}\in \mathbb{F}$,
  and defines the polynomial $p(x)=s+\sum^{t-1}_{i=1}a_ix^i$.
  This is a uniform degree-$(t-1)$ polynomial with constant term $s$.
  The share of user $id_i$ is $s_i=p(id_i)\in \mathbb{F}$.
  \item $\mathsf{Combine}((id_{i_1},s_{i_1}),\cdots,(id_{i_t},s_{i_t}))\rightarrow s$:
  On input
  any $t$ identities $id_{i_j},j=1,\cdots,t$,
  and associated shares $s_{i_j}, j=1,\cdots,t$. Using polynomial interpolation, it computes the unique degree-$(t-1)$
  polynomial $p$ for which $p(id_{i_j}) = s_{i_j}, j=1,\cdots,t$. The combining algorithm outputs the secret $s=p(0)$.
\end{itemize}
\begin{trivlist}
\item \textbf{Correctness.} It is clear that the combining algorithm works since
the secret $p(0)=s$ can be constructed from any $t$ shares. More precisely, 
by the Lagrange interpolation formula, given any $t$ points $(id_{i_j},p(id_{i_j}))$, $j=1,\cdots, t$,
$$
p(x)=\sum^t_{j=1}p(id_{i_j})\prod\limits_{l=1,l\not=j}^t\frac{x-id_{i_l}}{id_{i_j}-id_{i_l}},
$$
we can compute all points  
$(id_{i_v},p(id_{i_v}))$, $id_{i_v}\in \mathbb{F}$, where the secret is a special point $(0,s=p(0))$.
\end{trivlist}

\begin{trivlist}
\item \textbf{Security.} The sharing algorithm $\mathsf{Share}$ has perfect privacy, that is,
any $t-1$ users learn nothing of secret $s$ from their shares.
For any $t-1$ users corresponding to identities $id_{i_j}, j=1,\cdots,t-1$
and for any secret $s$ (namely, $p(0)$),
the distributions of $t-1$ shares of $s$ are perfectly indistinguishable from
$t-1$ independently uniform distributions.

In this paper, when building TLTDF
from the LWE assumption, 
we take advantage of the technique of
``clearing out the denominator''~\cite{Sho00,ABVVW12,BGGK17}
and the fact that the term $(n!)^2\cdot L_j$ is an integer,
where $L_j, j=1,\cdots, t$ are Lagrangian coefficients.
\end{trivlist}

\begin{lemma} (\cite{BGGK17}, Lemma 2.2).
For any $t$ identities $id_{i_j}=i_j,i_j\in [n],j=1,\cdots, t$,
the product $(n!)^2\cdot L_j$ is an integer,
and $|(n!)^2\cdot L_j|\leq (n!)^3$.
\end{lemma}

\subsection{Threshold Encryption}
We now recall the definition of TPKE from~\cite{Wee11}.
A $(n,t)$-TPKE consists of five algorithms as follows:
\begin{itemize}
  \item $\mathsf{Gen}(1^\lambda)\rightarrow (pk,msk)$: On input the security parameter $1^\lambda$,
  the key generation algorithm outputs a public key $pk$ and
  a master secret key $msk$.

  \item $\mathsf{Share}(msk,id_i)\rightarrow sk_{i}$:
  On input the master secret key $msk$ and a new identity $id$ associated with the user,
  the sharing algorithm outputs the shared secret key $sk_{i}$.

  \item $\mathsf{Enc}(pk,m)\rightarrow c$: On input the public key $pk$ and a message $m$,
  the encryption algorithm outputs a ciphertext $c$.

  \item $\mathsf{Dec}(sk_i,c)\rightarrow \delta_i$: On input a shared secret key $sk_i,i\in[n]$
  and the ciphertext $c$, the decryption algorithm outputs a
  decryption share $\delta_i$.

  \item $\mathsf{Combine}(\delta_{i_1},\cdots,\delta_{i_t},c)\rightarrow m$:
  On input any $t$ decryption shares $\delta_{i_j},j=1,\cdots,t$ and the ciphertext $c$,
  the combining algorithm outputs the message $m$.
\end{itemize}

\begin{trivlist}
\item \textbf{Correctness.} For any message $m$, $c\leftarrow \mathsf{Enc}(pk,m)$,
and any $t$ decryption shares $\delta_{i_1},\cdots,\delta_{i_t}$,
we have $\mathsf{Combine}(\delta_{i_1},\cdots,\delta_{i_t},c)=m$.
\end{trivlist}

\begin{trivlist}
\item \textbf{Security.} Let $\mathcal{A}$ be a PPT adversary against IND-CPA security of TPKE scheme
with adaptive corruption. Its advantage function 
$\mathbf{Adv}^{\text{ind-cpa}}_{\text{TPKE},\mathcal{A}}(\lambda)$ is defined as

\begin{equation*}
\begin{aligned}
\mathbf{Adv}^{\text{ind-cpa}}_{\text{TPKE},\mathcal{A}}(\lambda)=\left|\Pr\left[b=b':\begin{array}{l}
                  (pk,msk)\leftarrow \mathsf{Gen}(1^\lambda);  \\
                  (m_0, m_1)\leftarrow \mathcal{A}^{\mathsf{Share}(msk, \cdot)}(pk); \\
                   b\leftarrow \{0,1\}, c^* \leftarrow \mathsf{Enc}(pk,m_b); \\
                   b'\leftarrow \mathcal{A}^{\mathsf{Dec}(\cdot,\mathsf{Enc}(pk))}(pk,sk_{i_j},c^*);\\
                \end{array}
\right]-\frac{1}{2}\right|
\end{aligned}
\end{equation*}
Here, $\mathsf{Share}(msk, \cdot)$ denotes a shared secret key oracle that given an input of any identity $id$, and outputs a shared secret key $sk_{id}$. The adversary can query the oracle at most $t-1$ times adaptively .
$\mathsf{Dec}(\cdot,\mathsf{Enc}(pk))$ 
denotes an oracle that given an input of any identity $id$, 
computes a fresh ciphertext
$c$ using $\mathsf{Enc}(pk)$ and returns a decryption share $\mathsf{Dec}(sk_{id},c)$. 
This captures that 
the adversary may obtain decryption shares of fresh encryptions of known messages. 
The $(n,t)$-TPKE scheme is
IND-CPA secure, 
if for all PPT adversary the advantage function is negligible.

\end{trivlist}

\subsection{Revocation Encryption}
We recall the definition of RPKE from~\cite{NP00}.
A $(n,r)$-RPKE consists of four algorithms as follows:

\begin{itemize}
\item $\mathsf{Gen}(1^\lambda,r)\rightarrow (pk,msk)$:
  On input the security parameter $1^\lambda$, and the revocation threshold $r$,
  the key generation algorithm outputs a public key $pk$ and a master secret key $msk$.

  \item $\mathsf{Reg}(msk,id_i)\rightarrow sk_{i}$:
  On input the master secret key $msk$ and a new identity $id$ associated with the user,
  the registration algorithm outputs the shared secret key $sk_{i}$.

  \item $\mathsf{Enc}(pk,S,s)\rightarrow c$:
  On input the public key $pk$, a set $S$\footnote{The set $S$ contains the identities and shared secret keys of 
  revoked users.} 
  of revoked users (with $|S|\leq r$)
  and a session key $s$,
  the encryption algorithm outputs a ciphertext $c$.

  \item $\mathsf{Dec}(sk_i,c)\rightarrow s$:
  On input a shared secret key $sk_i$ of user $id_i$ and the ciphertext $c$, the decryption algorithm outputs
  the session key $s$, if $id_i$ is a legitimate user when $c$ is constructed.
\end{itemize}

\begin{trivlist}
\item \textbf{Correctness.} For any $id_i$, $i\in[n]$,
$(pk,msk)\leftarrow \mathsf{Gen}(1^\lambda)$, any $s$,
and any set $S$, $c\leftarrow \mathsf{Enc}(pk,S,s)$,
we require that for any non-revoked secret key $sk_i$, $s=\mathsf{Dec}(sk_i,c)$.
\end{trivlist}

\begin{trivlist}
\item \textbf{Security.}
Let $\mathcal{A}$ be a PPT adversary against IND-CPA security of RPKE scheme with adaptive corruption. 
Its advantage function $\textbf{Adv}^{\text{ind-cpa}}_{\text{RPKE},\mathcal{A}}(\lambda)$ is defined as

\begin{equation*}
\begin{aligned}
\textbf{Adv}^{\text{ind-cpa}}_{\text{RPKE},\mathcal{A}}(\lambda)=\left|\Pr\left[b=b':\begin{array}{l}
                  (pk,msk)\leftarrow \mathsf{Gen}(1^\lambda);  \\
                  (s_0, s_1)\leftarrow \mathcal{A}^{ \mathsf{Reg}(msk,\cdot)}(pk); \\
                   b\leftarrow \{0,1\}, c^* \leftarrow \mathsf{Enc}(pk,S,s_b); \\
                   b'\leftarrow \mathcal{A}(pk,sk_{i_j},c^*),j\in[r];\\
                \end{array}\right]-\frac{1}{2}\right|
\end{aligned}
\end{equation*}
Here, $\mathsf{Reg}(msk, \cdot)$ denotes an oracle that given an input of any identity $id$, and outputs a shared secret key $sk_{id}$. The adversary can query the oracle at most $r$ times adaptively.
If for all PPT adversary the advantage function is negligible, the $(n,r)$-RPKE scheme is IND-CPA secure. 

\end{trivlist}

\section{Threshold Trapdoor Function}\label{Threshold Trapdoor Function}
We give the definition and the security of TTDF as follows.
\begin{definition}
A collection of $(n,t)$-TTDFs is a tuple of polynomial-time algorithms
defined as follows:
\begin{itemize}
  \item $\mathsf{Gen}(1^\lambda)\rightarrow (ek,mtd)$:
  The generation algorithm is a probabilistic algorithm that 
  on input the security parameter $1^\lambda$,
  outputs a function index $ek$ and
  a master trapdoor $mtd$.
  \item $\mathsf{Share}(mtd,id_i)\rightarrow td_i$:
  The sharing algorithm is a deterministic algorithm that 
  on input the master trapdoor $mtd$ and any identity $id_i,i\in[n]$,
  outputs the shared trapdoor $td_i,i\in[n]$.
  \item $\mathsf{F}(ek,x)\rightarrow y$:
  On input the function index $ek$ and $x\in \{0,1\}^l$,
  the evaluation algorithm outputs $y$.
  \item $\mathsf{F}^{-1}(td_i,y)\rightarrow \delta_{i}$:
  On input any shared trapdoor $td_{i}$
  and an image $y$,
  the partial inversion algorithm outputs the inversion share $\delta_{i}$.
  \item $\mathsf{CombineF}^{-1}(ek,x,\delta_{i_1},\cdots,\delta_{i_{t-1}},id_{i_t})\rightarrow \delta_{i_t}$:
  On input $ek$, $x\in \{0,1\}^l$,   
  any $t-1$ inversion shares $\delta_{i_1},\cdots,\delta_{i_{t-1}}$ of the image of $x$, and identity $id_{i_t}$, 
  the combining inversion algorithm outputs the inversion share $\delta_{i_t}$ of identity $id_{i_t}$.
  \item $\mathsf{Combine}(\delta_{i_1},\cdots,\delta_{i_{t}},y)\rightarrow x$:
  On input any $t$ inversion shares $\delta_{i_j},j=1,\cdots,t$ and the image $y$,
  the combining algorithm outputs $x$.
\end{itemize}
\end{definition}
Note that we require that in the partial inversion algorithm and 
the combining algorithm, if a value $y$ is not in the image, the behavior of the algorithms are unspecified.

\begin{trivlist}
\item \textbf{Correctness.} For any $id_i$, $(ek, mtd)$$\leftarrow$$\mathsf{Gen}(1^\lambda)$,
$td_i$$\leftarrow$$\mathsf{Share}(mtd$, $id_i)$, $i\in[n]$,
$x\leftarrow \{0,1\}^l$, $y=\mathsf{F}(ek,x)$,
we require that 
for any $t$ shared trapdoors $td_{i_1}$, $\cdots$, $td_{i_t}$, we have 
$$x=\mathsf{Combine}(\mathsf{F}^{-1}(td_{i_1},y),\cdots,\mathsf{F}^{-1}(td_{i_t},y),y).$$
\end{trivlist}

\begin{trivlist}
\item \textbf{Security.}
Let $\mathcal{A}$ be a PPT adversary against $(n,t)$-TTDF
and define its advantage function $\textbf{Adv}^{\text{tow}}_{\text{TTDF},\mathcal{A}}(\lambda)$ as

\begin{equation*}
\begin{aligned}
\textbf{Adv}^{\text{tow}}_{\text{TTDF},\mathcal{A}}(\lambda)=\Pr\left[x=x':\begin{array}{l}
                  (ek,mtd)\leftarrow \mathsf{Gen}(1^\lambda);  \\
                  x\leftarrow \{0,1\}^l, y=\mathsf{F}(ek,x); \\
                  x'\leftarrow \mathcal{A}^{\mathsf{Share}(mtd, \cdot)}(ek,y)\\
                \end{array}
\right]
\end{aligned}
\end{equation*}
Here, $\mathsf{Share}(mtd, \cdot)$ denotes a shared trapdoor oracle that given an input of any identity $id$, and outputs a shared trapdoor $td_{id}$. The adversary can query the oracle at most $t-1$ times adaptively.
If for any PPT adversary the advantage function is negligible, 
$(n,t)$-TTDF is threshold one-way.
\end{trivlist}

\subsection{Connection to Function Sharing}
De Santis et al.~\cite{DDFY94} introduced the notion of function sharing (FS) parameterized by 
the threshold value $t$ and the number of identities $n$. 
$(n,t)$-FS can split the master trapdoor into $n$ shared trapdoors, 
where $n$ is a fixed polynomial of the security parameter.
The function is easy to invert when given
$t$ shared trapdoors, 
while any PPT adversary
cannot invert the function even if it obtains any $t-1$ shared trapdoors
and a history tape $H$ that contains partial inversion shares of polynomial many random 
images. Then they constructed threshold cryptosystems based on FS and 
instantiated it under the RSA assumption. 
However, the number of identities of their FS and TPKE 
is limited in a fixed polynomial of security parameter. 

Our notion of TTDF differs from FS as follows: 
TTDF can support exponential number of identities and 
the generation algorithm and the sharing algorithm of TTDF are 
independent of the number of identities. Moreover, 
in the security experiment of TTDF, it omits the complicated 
history tape defined in the security experiment of FS. 
More precisely, 
$(n,t)$-TTDF has an additional combining inversion algorithm that  
given the function index $ek$, 
any preimage $x$ and any $t-1$ inversion shares of the image of $x$,
can compute the inversion share of any other identity. 
In the security experiment of $(n,t)$-TTDF, 
any adversary given any $t-1$ shared trapdoors 
can compute inversion shares of any other identity with the help of the preimage. 
While in the security experiment of $(n,t)$-FS, 
the adversary can only look up the history tape $H$ to 
obtain inversion shares of some identities contained in $H$, and 
the length of $H$ is a fixed polynomial of the security parameter.  
Therefore, $(n,t)$-TTDF implies $(n,t)$-FS, and 
$(n,t)$-TTDF can be used to construct TPKE scheme~\cite{Wee11} which supports ad-hoc groups 
(i.e., exponential number of identities and the generation algorithm is independent of 
the number of identities), 
the reason is that the reduction algorithm who holds any $t-1$ shared trapdoors can answer the 
shared decryption oracle of all identities.

\section{Threshold Encryption from TTDF}\label{Threshold Encryption from TTDF}

Let $(\mathsf{Gen},\mathsf{Share},\mathsf{F},
\mathsf{F}^{-1},\mathsf{CombineF}^{-1},\mathsf{Combine})$ be a $(n,t)$-TTDF and
$\mathsf{hc}(\cdot)$ be a hardcore function. We construct a TPKE as follows:

\begin{itemize}
  \item $\mathsf{Gen}(1^\lambda)\rightarrow (pk,msk)$:
  On input the security parameter $1^\lambda$, 
  the generation algorithm runs $(ek,mtd)\leftarrow \text{TTDF.}\mathsf{Gen}(1^\lambda)$
  and outputs a public key $pk=ek$ and a master secret key $msk=mtd$.

  \item $\mathsf{Share}(msk,id_i)\rightarrow sk_i$:
  On input the master secret key $msk$ and any identity $id_i,i\in[n]$,
  the sharing algorithm runs $td_i\leftarrow \text{TTDF.}\mathsf{Share}(msk,id_i)$ and
  outputs the shared secret key $sk_i=td_i,i\in[n]$.

  \item $\mathsf{Enc}(pk,m)\rightarrow c$:
  On input the public key $pk$ and a message $m$,
  the encryption algorithm chooses $x\leftarrow \{0,1\}^l$,
  computes $c_1=\text{TTDF.}\mathsf{F}(pk,x)$, $c_2=\mathsf{hc}(x)\oplus m$,
  and outputs the ciphertext $c=(c_1,c_2)$.

  \item $\mathsf{Dec}(sk_i,c)\rightarrow \delta_i$:
  On input a secret key $sk_i$ and a ciphertext $c$, the decryption algorithm computes
  $\delta_i=\text{TTDF.}\mathsf{F}^{-1}(sk_i,c_1)$, and outputs a decryption share $\delta_i$.

  \item $\mathsf{Combine}(\delta_{i_1},\cdots,\delta_{i_t},c)\rightarrow m$:
  On input any $t$ decryption shares $\delta_{i_j},j=1,\cdots,t$ and the ciphertext $c=(c_1,c_2)$,
  the combining algorithm computes $x=\text{TTDF.}\mathsf{Combine}(\delta_{i_1},\cdots,\delta_{i_t},c_1)$,
  $m=\mathsf{hc}(x)\oplus c_2$.
  It outputs the message $m$.
\end{itemize}

\begin{theorem}
If the TTDF is threshold one-way, then the TPKE is IND-CPA secure.
\end{theorem}

\begin{proof}
We define two hybrid experiments $\text{Game}_1$, $\text{Game}_2$.
\begin{itemize}
\item $\text{Game}_1$: The game is identical to the IND-CPA experiment. At the beginning,
the challenger runs $\mathsf{Gen}$
to obtain $pk$ and $msk$.
The challenger sends $pk$ to $\mathcal{A}$.
$\mathcal{A}$ can query the shared secret key oracle $\mathsf{Share}(msk,\cdot)$ adaptively.
Then the challenger runs the sharing algorithm 
to answer $\mathcal{A}$.
Upon receiving the messages $m_0$, $m_1$ from $\mathcal{A}$, 
the challenger chooses $b\in\{0,1\}$ at random and
returns $c^*=\mathsf{Enc}(pk,m_b)$ to $\mathcal{A}$.
$\mathcal{A}$ is still able to have access to 
the oracle $\mathsf{Dec}(\cdot,\mathsf{Enc}(pk))$.
At the end of the game, $\mathcal{A}$ outputs $b'\in \{0,1\}$ as the guess of $b$.
If $b'=b$, $\mathcal{A}$ wins this game, otherwise fails.

\item $\text{Game}_2$: The game is identical to 
$\text{Game}_1$, except that when the challenger generates the
challenge ciphertext $c^*=(c_1,c_2)$, it replaces $c_2=m_b\oplus \mathsf{hc}(x)$
with $c_2=m_b\oplus r$.
\end{itemize}
For $i\in \{1,2\}$, let $\Pr[\mathcal{A}^{\text{Game}_i}=b]$ be the probability
that $\mathcal{A}$ outputs the
bit $b$ when executed in $\text{Game}_i$.
We claim that if there is an adversary $\mathcal{A}$ against
the TPKE
such that $\Pr[\mathcal{A}^{\text{Game}_1}=b]
-\Pr[\mathcal{A}^{\text{Game}_2}=b]$ is non-negligible, we can
construct a distinguisher $\mathcal{D}$ against the hardcore function.
On input $(ek,y,r)$,
where $ek$ is a function index,
$y=\mathsf{F}(ek,x)$ with $x\leftarrow\{0,1\}^l$ and $r$ is either $\mathsf{hc}(x)$ or
a random string, $\mathcal{D}$ works as follows:

\begin{enumerate}
  \item $\mathcal{D}$ runs $\mathcal{A}$ on input $pk=ek$.$\mathcal{D}$ can simulate the shared secret key oracle by querying the shared trapdoor oracle $\mathsf{Share}(mtd,\cdot)$ adaptively.

  \item Upon receiving two messages $m_0$, $m_1$ from $\mathcal{A}$, 
  $\mathcal{D}$ chooses $b\in \{0,1\}$ at random, let $c_1=y$, $c_2=m_b\oplus r$,
  and returns $c=(c_1,c_2)$ to $\mathcal{A}$. 
  $\mathcal{A}$ is able to have access to the oracle $\mathsf{Dec}(\cdot,\mathsf{Enc}(pk))$. 
  $\mathcal{D}$ chooses $x'$\footnote{The preimage $x'$ is chosen by $\mathcal{D}$, so it can computes decryption share of any identity.} in the domain at random, fixes the corrupted secret keys, and computes the decryption share by the combining inversion algorithm to simulate the oracle $\mathsf{Dec}(\cdot,\mathsf{Enc}(pk))$. More precisely, $\mathcal{D}$ computes 
  $c'_1=\text{TTDF.}\mathsf{F}(pk,x')$, 
  $\delta_{i_1}=$TTDF.$\mathsf{F}^{-1}(td_{i_1},c_1)$, $\cdots$, 
  $\delta_{i_{t-1}}=$TTDF.$\mathsf{F}^{-1}(td_{i_{t-1}}$, $c_1)$,
  $\delta'=$TTDF.$\mathsf{CombineF}^{-1}($
  $x'$, $c_1$, $\delta_{i_1}$, 
  $\cdots$, $\delta_{i_{t-1}}$, $id')$,
  then $\mathcal{D}$ returns $\delta'$ to $\mathcal{A}$.
  At last $\mathcal{D}$ outputs what $\mathcal{A}$ outputs.

  \item if $b=b'$, $\mathcal{D}$ returns ``1'' to denote $r$ is the output of the hardcore function,
  otherwise returns ``0'' to denote $r$ is a random string.
\end{enumerate}
The distinguisher $\mathcal{D}$ can give a perfect simulation of either $\text{Game}_1$ 
or $\text{Game}_2$.
The advantage of $\mathcal{D}$ is non-negligible, which is a contradiction of the threshold one-wayness.
Therefore,
$|\Pr[\mathcal{A}^{\text{Game}_1}=b]-
\Pr[\mathcal{A}^{\text{Game}_2}=b]|\leq \mathsf{negl}(\lambda)$.

Finally, in $\text{Game}_2$
the output of hardcore function has been replaced with a random string, 
so $\Pr[\mathcal{A}^{\text{Game}_2}=b]=1/2$.
We have:

\begin{equation*}
\begin{aligned}
\Pr[\mathcal{A}^{\text{Game}_1}=b]&\leq |\Pr[\mathcal{A}^{\text{Game}_1}=b]-\Pr[\mathcal{A}^{\text{Game}_2}=b]|+\Pr[\mathcal{A}^{\text{Game}_2}=b]\\
&\leq \frac{1}{2}+\mathsf{negl}(\lambda)
\end{aligned}
\end{equation*}
Therefore, the TPKE is IND-CPA secure.
\end{proof}

\section{Revocation Encryption from TTDF}\label{Revocation Encryption from TTDF}

Let $(\mathsf{Gen},\mathsf{Share},\mathsf{F},
\mathsf{F}^{-1},\mathsf{CombineF}^{-1},\mathsf{Combine})$ be a $(n,t)$-TTDF and
$\mathsf{hc}(\cdot)$ be a hardcore function.
We construct a $(n,t-1)$-RPKE as follows:

\begin{itemize}
\item $\mathsf{Gen}(1^\lambda)\rightarrow (pk,msk)$:
  On input the security parameter $1^\lambda$,
  the generation algorithm
  runs $(ek,mtd)\leftarrow\text{TTDF.}\mathsf{Gen}(1^\lambda)$ and
  outputs a public key $pk=ek$ and a master secret key $msk=mtd$.

  \item $\mathsf{Reg}(msk,id_i)\rightarrow sk_i$:
  On input the master secret key $msk$ and any identity $id_i,i\in[n]$,
  the registration algorithm
  runs $td_i\leftarrow\text{TTDF.}\mathsf{Share}(mtd,id_i)$ and
  outputs the shared secret key $sk_i=td_i,i\in[n]$.

  \item $\mathsf{Enc}(pk,sk_{i_1},\cdots,sk_{i_{t-1}},s)\rightarrow c$:
  On inputs the public key $pk$, a set of $t-1$ revoked secret keys
  $sk_{i_j},j=1,\cdots,t-1$
  and a session key $s$, the encryption algorithm chooses $x\leftarrow \{0,1\}^l$,
  computes $c_1=\text{TTDF.}\mathsf{F}(pk,x)$, $c_2=\mathsf{hc}(x)\oplus s$ and
  $\delta_{i_j}=\text{TTDF.}\mathsf{F}^{-1}(sk_{i_j},c_1),j=1,\cdots,t-1$.
  It outputs the ciphertext $c=(c_1,c_2,\delta_{i_1},\cdots,\delta_{i_{t-1}})$.

  \item $\mathsf{Dec}(sk_{i_j},c)\rightarrow s$:
  On inputs a secret key $sk_{i_j},j\not=1,\cdots,t-1$ and a ciphertext $c$, the decryption algorithm computes
  $\delta_{i_j}=\text{TTDF.}\mathsf{F}^{-1}(sk_{i_j},c_1),j\not=1,\cdots,t-1$,
  $x=\text{TTDF.}\mathsf{Combine}(\delta_{i_1}$, $\cdots$, $\delta_{i_{t-1}}$, $\delta_{i_j}$, $c_1)$
  and $s=\mathsf{hc}(x)\oplus c_2$.
  It outputs session key $s$.
\end{itemize}
\begin{theorem}
If the TTDF is threshold one-way, then the RPKE is IND-CPA secure.
\end{theorem}
\begin{proof}
We define two hybrid experiments $\text{Game}_1$, $\text{Game}_2$.
\begin{itemize}
  \item $\text{Game}_1$: The game is identical to the IND-CPA experiment.
  At the beginning, the challenger runs
  $(pk,$ $msk)$ $\leftarrow \mathsf{Gen}(1^\lambda)$ and gives the
  $pk$ to the adversary $\mathcal{A}$,
  $\mathcal{A}$ can query the shared secret key oracle $\mathsf{Reg}(msk,\cdot)$ adaptively.
  Then the challenger runs the registration algorithm to answer $\mathcal{A}$.
  Upon receiving two session keys $s_0$, $s_1$ from $\mathcal{A}$, 
  the challenger chooses $b\in\{0,1\}$ at random and returns
  $c^*=\mathsf{Enc}(pk,sk_{i_1},\cdots,sk_{i_{t-1}},s_b)$ to $\mathcal{A}$.
  At the end of the game, $\mathcal{A}$ outputs $b'\in \{0,1\}$ as the guess of $b$.
  If $b'=b$, $\mathcal{A}$ wins this game, otherwise fails.

  \item $\text{Game}_2$: The game is identical to $\text{Game}_1$,
  except when the challenger generates the challenge ciphertext
  $c^*=(c_1,c_2,\delta_{i_1},\cdots,\delta_{i_{t-1}})$, it replaces
  $c_2=s_b\oplus \mathsf{hc}(x)$ with $c_2=s_b\oplus r$,
  where $r$ is a random string.
\end{itemize}
For $i\in \{1,2\}$, let $\Pr[\mathcal{A}^{\text{Game}_i}=b]$ be the probability
that $\mathcal{A}$ outputs the
bit $b$ when executed in $\text{Game}_i$. We claim that if there is an adversary $\mathcal{A}$
such that $\Pr[\mathcal{A}^{\text{Game}_1}=b]-
\Pr[\mathcal{A}^{\text{Game}_2}=b]$ is non-negligible, we can
construct a distinguisher $\mathcal{D}$ against the hardcore function.
On input $(ek,y,r)$,
where $ek$ is a function index,
$y=\mathsf{F}(ek,x)$ with $x\leftarrow\{0,1\}^l$ and $r$ is either $\mathsf{hc}(x)$ or
a random string, $\mathcal{D}$ works as follows:
\begin{enumerate}
  \item $\mathcal{D}$ runs $\mathcal{A}$ on input $pk=ek$. $\mathcal{D}$ can simulate the shared secret key oracle by querying the shared trapdoor oracle $\mathsf{Share}(mtd,\cdot)$ adaptively.

  \item Upon receiving two session keys $s_0$, $s_1$. $\mathcal{D}$ chooses $b\in \{0,1\}$ at random, let $c_1=y$, $c_2=s_b\oplus r$,
  computes $\delta_{i_j}=\mathsf{F}^{-1}(sk_{i_j},c_1),j=1,\cdots,t-1$,
  returns $c=(c_1,c_2,\delta_{i_1},\cdots,\delta_{i_{t-1}})$ to $\mathcal{A}$
  and gets a bit $b'$ output by $\mathcal{A}$.
  \item if $b=b'$, $\mathcal{D}$ returns ``1'' to denote $r$ is the output of the hardcore function,
  otherwise returns ``0'' to denote $r$ is a random string.
\end{enumerate}
The distinguisher $\mathcal{D}$ can give a perfect simulation of either 
$\text{Game}_1$ or $\text{Game}_2$.
The advantage of $\mathcal{D}$ is non-negligible, which is a contradiction of the threshold one-wayness.
Therefore,
$|\Pr[\mathcal{A}^{\text{Game}_1}=b]-\Pr[\mathcal{A}^{\text{Game}_2}=b]|\leq \mathsf{negl}(\lambda)$.

Finally, in $\text{Game}_2$ the output of hardcore function has been replaced with a random string, 
so $\Pr[\mathcal{A}^{\text{Game}_2}=b]=1/2$.
We have:
\begin{equation*}
\begin{aligned}
\Pr[\mathcal{A}^{\text{Game}_1}=b]&\leq|\Pr[\mathcal{A}^{\text{Game}_1}=b]-\Pr[\mathcal{A}^{\text{Game}_2}=b]|+\Pr[\mathcal{A}^{\text{Game}_2}=b]\\
&\leq \frac{1}{2}+\mathsf{negl}(\lambda)
\end{aligned}
\end{equation*}
Therefore, the RPKE is IND-CPA secure.
\end{proof}

\section{Threshold Lossy Trapdoor Function}\label{Threshold Lossy Trapdoor Function}
In this section, we  introduce a new cryptographic primitive called TLTDF which is a
threshold version of LTDF~\cite{PW08} and prove that TLTDF implies TTDF.

Let $l(\lambda)=\mathsf{poly}(\lambda)$ denote the input length of the function,
$k(\lambda)\leq l(\lambda)$ and $r(\lambda)=l(\lambda)-k(\lambda)$ denote
the lossiness and the residual leakage.
For notational convenience, we often omit the dependence on $\lambda$, and
define the sampling algorithm $\mathsf{Samp}_{\textnormal{inj}}(\cdot):=$$\mathsf{Samp}(\cdot,1)$ samples injective mode
and $\mathsf{Samp}_{\textnormal{loss}}(\cdot):=$$\mathsf{Samp}(\cdot,0)$ samples lossy mode.
\begin{definition}
A collection of $(n,t,l,k)$-TLTDFs is a tuple of polynomial-time algorithms
defined as follows. 
\begin{itemize}
  \item $\mathsf{Samp}_{\textnormal{inj}}(1^\lambda)\rightarrow (ek,mtd)$:
  The sampling algorithm is a probabilistic algorithm that 
  on input the security parameter $1^\lambda$, 
  outputs a function index $ek$ and a master trapdoor $mtd$.
  \item $\mathsf{Samp}_{\textnormal{loss}}(1^\lambda)\rightarrow (ek,mtd)$:
  The sampling algorithm is a probabilistic algorithm that 
  on input the security parameter $1^\lambda$, 
  outputs a function index $ek$ and a master trapdoor $mtd$.
  \item $\mathsf{Share}(mtd,id_i)\rightarrow td_i$:
  The sharing algorithm is a deterministic algorithm that 
  on input the master trapdoor $mtd$ and any identity $id_i,i\in[n]$,
  outputs the shared trapdoor $td_i,i\in[n]$ in both modes.
  \item $\mathsf{F}(ek,x)\rightarrow y$:
  On input the function index $ek$
  and $x\in \{0,1\}^l$, 
  the evaluation algorithm outputs $y$ in both modes,
  but the image has size at most $2^r=2^{l-k}$ in the lossy mode.
  \item $\mathsf{F}^{-1}(td_i,y)\rightarrow \delta_{i}$:
  On input any shared trapdoor $td_{i}$, $i\in[n]$
  and an image $y$, the partial inversion algorithm outputs an inversion share $\delta_i$.
  \item $\mathsf{CombineF}^{-1}(ek,x,\delta_{i_1},\cdots,\delta_{i_{t-1}},id_{i_t})\rightarrow \delta_{i_t}$:
  On input $ek$, $x\in \{0,1\}^l$,  
  any $t-1$ inversion shares $\delta_{i_1},\cdots,\delta_{i_{t-1}}$ of the image of $x$, and identity $id_{i_t}$, 
  the combining inversion algorithm outputs the inversion share $\delta_{i_t}$ of identity $id_{i_t}$.  
  \item $\mathsf{Combine}(\delta_{i_1},\cdots,\delta_{i_{t}},y)\rightarrow x$:
  On input any $t$ inversion shares $\delta_{i_j},j=1,\cdots,t$ and the image $y$, in injective mode
  the combining algorithm outputs $x$.
\end{itemize}
\end{definition}
Note that  
we require that the shared trapdoors in both modes have the same space,
and the behavior of the partial inversion algorithm and the combining algorithm
is unspecified, if a value $y$ is not in the image.

\begin{trivlist}
\item \textbf{Security.}
Let $\mathcal{A}$ be a PPT adversary against TLTDF and define its advantage function as

\begin{equation*}
\begin{aligned}
\textbf{Adv}^{\text{ind}}_{\text{TLTDF},\mathcal{A}}(\lambda)=\left|\Pr\left[b=b':\begin{array}{l}
                  b\leftarrow \{0,1\};\\
                  (ek,mtd)\leftarrow \mathsf{Samp}(1^\lambda,b);  \\
                   b'\leftarrow \mathcal{A}^{\mathsf{Share}(mtd,\cdot)}(ek)\\
                \end{array}
\right]-\frac{1}{2}\right|
\end{aligned}
\end{equation*}
Here, $\mathsf{Share}(mtd, \cdot)$ denotes a shared trapdoor oracle that given an input of any identity $id$, and outputs a shared trapdoor $td_{id}$. The adversary can query the oracle at most $t-1$ times adaptively.
A TLTDF is said to produce indistinguishable function indexes 
if $\textbf{Adv}^{\text{ind}}_{\text{TLTDF},\mathcal{A}}(\lambda)$ 
is negligible for all adversary.
\end{trivlist}

\begin{theorem}
If the sharing algorithm holds perfect privacy
and the injective and lossy modes of LTDF are indistinguishable,
then the TLTDF described above is also hard to distinguish injective from lossy.
\end{theorem}

\begin{proof}
We define four hybrid experiments 
$\text{Game}_1$, $\text{Game}_2$, $\text{Game}_3$, $\text{Game}_4$.
\begin{itemize}
\item $\text{Game}_1$: The challenger runs
$(ek, mtd)\leftarrow \mathsf{Samp}_{\textnormal{inj}}(1^\lambda)$, and gives $ek$
to $\mathcal{A}$.
$\mathcal{A}$ can adaptively query the shared trapdoor oracle $\mathsf{Share}(mtd,\cdot)$ at most $t-1$ times.
Then the challenger runs the sharing algorithm to answer $\mathcal{A}$.

\item $\text{Game}_2$: The game is identical to $\text{Game}_1$,
except that the challenger generates the corrupted trapdoors
by choosing $r_i, i=1,\cdots,t-1$ at random in the shared trapdoor space
and then gives them to $\mathcal{A}$.

\item $\text{Game}_3$: The game is identical to $\text{Game}_2$,
except that the challenger runs $(ek, mtd)\leftarrow \mathsf{Samp}_{\textnormal{loss}}(1^\lambda)$
instead of running $(ek, mtd)\leftarrow \mathsf{Samp}_{\textnormal{inj}}(1^\lambda)$.

\item $\text{Game}_4$: The game is identical to $\text{Game}_3$,
except that the challenger runs the sharing algorithm
to generate the shared trapdoor and gives the shared trapdoors to $\mathcal{A}$.
\end{itemize}
The adversary's view is perfectly indistinguishable in $\text{Game}_1$ and 
$\text{Game}_2$ with the replacement of the shared trapdoors, since the sharing algorithm has
the perfect privacy.
Similarly, the adversary's view is perfectly indistinguishable in $\text{Game}_3$ 
and $\text{Game}_4$.
The only difference between $\text{Game}_2$ and 
$\text{Game}_3$ is the sampling algorithm. So
the adversary's view is computationally indistinguishable in 
$\text{Game}_2$ and $\text{Game}_3$,
the fact follows that the injective and lossy modes of LTDF are indistinguishable~\cite{PW08}.
Therefore, the adversary's view is computationally indistinguishable in 
$\text{Game}_1$ and $\text{Game}_4$
and the TLTDF described above is hard to distinguish injective from lossy, even if the adversary can obtains any $t-1$ shared trapdoors adaptively.
\end{proof}

\begin{theorem}
Let $\text{TLTDF}=(\mathsf{Samp}_{\textnormal{inj}}$, 
$\mathsf{Samp}_{\textnormal{loss}}$, $\mathsf{Share}$, 
$\mathsf{F}$, $\mathsf{F}^{-1}$, $\mathsf{CombineF}^{-1}$, $\mathsf{Combine})$
give a collection of $(n,t,l,k)$-TLTDFs with
$k=\omega(\log \lambda)$. Then
$\text{TTDF}=(\mathsf{Samp}_{\textnormal{inj}}$, $\mathsf{Share}$,
$\mathsf{F}$, $\mathsf{F}^{-1}$, $\mathsf{CombineF}^{-1}$, $\mathsf{Combine})$
give a collection of $(n,t)$-TTDFs.
\end{theorem}
\begin{proof}
By definition, for any $id_i$, $i\in[n]$, 
$(ek$, $mtd)$$\leftarrow$$\mathsf{Samp}_{\textnormal{inj}}$$(1^\lambda)$,
$td_i\leftarrow \mathsf{Share}(mtd,id_i)$,
$x\leftarrow \{0,1\}^l$, $y=\mathsf{F}(ek,x)$, 
for any $t-1$ inversion shares 
$\delta_{i_1}=\mathsf{F}^{-1}(td_{i_1},y),\cdots, \delta_{i_{t-1}}=\mathsf{F}^{-1}(td_{i_{t-1}},y)$ 
and identity $id_{i_t}$, we have 
$\delta_{i_t}=\mathsf{F}^{-1}(td_{i_t},y)=\mathsf{CombineF}^{-1}(x$, $y$, 
$\delta_{i_1},\cdots,\delta_{i_{t-1}},id_{i_t}),$
and for any $t$ shared trapdoors $td_{i_1}$, $\cdots$, $td_{i_t}$, we have 
$x=\mathsf{Combine}(\mathsf{F}^{-1}(td_{i_1},y),\cdots,\mathsf{F}^{-1}(td_{i_t},y),y).$
Therefore, the correctness condition holds.
We prove that the function
also holds the threshold one-wayness:

Suppose $\mathcal{A}$
is a PPT inverter, if $\mathcal{A}$ can break the threshold one-wayness
with non-negligible probability, we can build
an adaptive distinguisher $\mathcal{D}$ between injective modes
and lossy ones.
$\mathcal{D}$ is given a function index $ek$ as input.
Its goal is to distinguish $ek$ is generated in the injective or lossy mode.
$\mathcal{D}$ works as follows:
\begin{enumerate}
  \item $\mathcal{D}$ runs inverter $\mathcal{A}$ on input the function index $ek$
  and gets identities output by $\mathcal{A}$.
  \item $\mathcal{D}$ chooses these identities to corrupt and
  obtains associated trapdoors,
  then $\mathcal{D}$ chooses $x\leftarrow \{0,1\}^l$, computes $y=\mathsf{F}(ek,x)$,
  gives the value $y$
  and the associated trapdoors to $\mathcal{A}$,
  and then obtains the value $x'$
  output by $\mathcal{A}$.
  \item if $x'=x$, $\mathcal{D}$ returns ``1'' to denote $ek$ is generated in the injective mode,
  otherwise returns ``0'' to denote $ek$ is generated in the lossy mode.
\end{enumerate}
First, by the assumption on $\mathcal{A}$,
if $ek$ is generated by $\mathsf{Samp}_{\textnormal{inj}}(1^\lambda)$,
we have $x'=x$ with
non-negligible probability and $\mathcal{D}$ outputs ``1''.
Suppose $ek$ is generated by $\mathsf{Samp}_{\textnormal{loss}}(1^\lambda)$.
The probability that
even an unbounded algorithm $\mathcal{A}$ predicts $x$ is given
by the average min-entropy of $x$ conditioned on $(ek,td_{i_1},\cdots,td_{i_{t-1}}$, $\mathsf{F}(ek,\cdot))$,
Because $\mathsf{F}(ek,\cdot)$ takes at most $2^{l-k}$ values, $ek$ and $x$ are independent.
By (\cite{PW08}, Lemma 2.1)
\begin{equation*}
\begin{aligned}
\widetilde{H}_{\infty}(x|ek,td_{i_1},\cdots,td_{i_{t-1}},\mathsf{F}(ek,x))&\geq\widetilde{H}_{\infty}(x|ek,td_{i_1},\cdots,td_{i_{t-1}})-(l-k)\\
&=l-(l-k)=k
\end{aligned}
\end{equation*}
where by the perfect privacy of the sharing algorithm, $td_{i_j},j=1,\cdots,t-1$ look like random numbers.
Since $k=\omega(\log \lambda)$,
the probability that $\mathcal{A}$ outputs $x$ and $\mathcal{D}$ outputs ``0''
is $\mathsf{negl}(\lambda)$. $\mathcal{D}$ distinguishes injective mode from lossy mode,
a contradiction of the hypothesis.
\end{proof}

\begin{trivlist}
\item \textbf{Remarks.}
In our applications of TPKE and RPKE, we use the pairwise independent hash function~\cite{WC81} as a hardcore function.
Let $H:\{0,1\}^l \rightarrow \{0,1\}^{l'}$ be a family of pairwise independent hash functions,
where $l'\leq k-2\lg(1/\epsilon)$ for some negligible $\epsilon=\mathsf{negl}(\lambda)$,
and we choose $\mathsf{hc}\leftarrow H$.
Following the Theorem 4,
$\widetilde{H}_{\infty}(x|ek,td_{i_1},\cdots,td_{i_{t-1}},\mathsf{F}(ek,x))\geq k$.
By the hypothesis that $l'\leq k-2\lg(1/\epsilon)$ and Lemma 2,
we have that $\mathsf{hc}(x)$ is $\epsilon$-close to uniform.
\end{trivlist}

\section{Instantiations of TLTDF}\label{Instantiations of TLTDF}

In this section, we give instantiations of TLTDF based on the DDH assumption and the LWE assumption.

\subsection{Instantiation of TLTDF Based on the DDH Assumption}\label{TLTDF Based on DDH}
By using the ElGamal-like encryption primitive in~\cite{PW08},
we generate a ciphertext $\mathbf{C}_1$ by encrypting the identity matrix $\mathbf{I}$ in the injective mode
and generate a ciphertext $\mathbf{C}_0$ by encrypting the all-zeros matrix $\mathbf{0}$ in the lossy mode.
\begin{lemma}(\cite{PW08}, Lemma 5.1).
The matrix encryption scheme produces indistinguishable ciphertexts under the DDH assumption.
\end{lemma}

\begin{trivlist}
\item\textbf{Construction.}
We now describe a DDH-based TLTDF as follows.
The identity space is given by $\mathbb{Z}_p\backslash \{0\}$.
\end{trivlist}

\begin{itemize}
  \item $\mathsf{Samp}_{\textnormal{inj}}$: On input $1^\lambda$, it chooses
  $(p, \mathbb{G}, g)\leftarrow$ $\mathsf{Gen}(1^\lambda)$,
  samples $r_i$, $s_i$, $b_{ij}\leftarrow \mathbb{Z}_p,i=1,\cdots,l$, $j=1,\cdots,t-1$ and computes
\begin{equation*}
\begin{aligned}
\mathbf{C}_1=\left(\begin{array}{cccc}
      g^{r_1} & g^{r_1s_1}g & \cdots & g^{r_1s_l} \\
      g^{r_2} & g^{r_2s_1} & \cdots & g^{r_2s_l} \\
       \vdots &     \vdots & \ddots & \vdots  \\
       g^{r_l}&  g^{r_ls_1}  & \cdots& g^{r_ls_l}g
    \end{array}\right)
\end{aligned}
\end{equation*}
  The function index is $ek=\mathbf{C}_1$
  and the master trapdoor is $mtd=((s_i),\mathbf{D}=(b_{ij}))$.

  \item $\mathsf{Samp}_{\textnormal{loss}}$:
  On input $1^\lambda$, it chooses $(p, \mathbb{G}, g)\leftarrow$ $\mathsf{Gen}(1^\lambda)$,
  samples $r_i$, $s_i$, $b_{ij}\leftarrow \mathbb{Z}_p,i=1,\cdots,l$, $j=1,\cdots,t-1$ and computes
\begin{equation*}
\begin{aligned}
\mathbf{C}_0=\left(\begin{array}{cccc}
      g^{r_1} & g^{r_1s_1} & \cdots & g^{r_1s_l} \\
      g^{r_2} & g^{r_2s_1} & \cdots & g^{r_2s_l} \\
       \vdots &     \vdots & \ddots & \vdots  \\
       g^{r_l}&  g^{r_ls_1}  & \cdots& g^{r_ls_l}
    \end{array}\right)
\end{aligned}
\end{equation*}
  The function index is $ek=\mathbf{C}_0$
  and the master trapdoor is $mtd=((s_i),\mathbf{D}=(b_{ij}))$.

  \item $\mathsf{Share}$: On input the master trapdoor
  $mtd$ and any identity $id_i,i=1,\cdots,n$,
  it sets 
$$
  f_j(x)=s_j+b_{j1}x+\cdots+b_{j(t-1)}x^{t-1},j\in[l].
$$
  and outputs
  $td_i^T=(f_1(id_i),\cdots, f_l(id_i))$.

  \item $\mathsf{F}$: On input a function index $ek$$=$$(c_{ij})_{l\times (l+1)}$
  and $\mathbf{x}\in \{0,1\}^l$, $\mathbf{x}=(x_1,\cdots,x_l)$, 
  it outputs
  $\mathbf{y}$$=$$(y_1$, $\cdots$, $y_{(l+1)})$, $y_i=c_{1i}^{x_1}c_{2i}^{x_2}\cdots c_{li}^{x_l},\ i=1,\cdots,l+1$.

  \item $\mathsf{F}^{-1}$: On input any shared trapdoor $td_i$ and the value $y_1$. It outputs
  $\delta_i^T=(y_1^{f_1(id_i)},\cdots,y_1^{f_l(id_i)})$.

  \item $\mathsf{CombineF}^{-1}$: On input $ek$, $\mathbf{x}\in \{0,1\}^l$, 
  any $t-1$ inversion shares $\delta_{i_j},j=1,\cdots,t-1$ 
  and identity $id_{i_t}$. Because of
  $f_{j}(id_{i_t})=\sum_{v=0}^{t-1}L_v f_{j}(id_{i_v})$, $j=1,\cdots,l$, where 
  $L_v, v=0,1,\cdots,t-1$
  are the Lagrangian coefficients which may be efficiently computed given
  $(id_{i_0}=0,id_{i_1},\cdots,id_{i_{t-1}})$, it computes $\mathbf{y}=\mathsf{F}(ek,\mathbf{x})$, and
  $y_1^{f_i(id_{i_0})}=y_1^{s_i}=y_{i+1}/x_i$,
  $y_1^{f_i(id_{i_t})}=\prod_{v=0}^{t-1}\left(y_1^{f_i(id_{i_v})}\right)^{L_v}$,
  $i=1,\cdots, l$, 
  and outputs $\delta_{i_t}^T=(y_1^{f_1(id_{i_t})}$, $\cdots$, $y_1^{f_l(id_{i_t})})$.

  \item $\mathsf{Combine}$: On input any $t$ inversion shares
  $\delta_{i_j},j=1,\cdots,t$ and the value
  $\mathbf{y}$. Because of
  $f_{j}(0)=\sum_{v=1}^{t}L_vf_{j}(id_{i_v})$, $j=1,\cdots,l$, where $L_v, v=1,\cdots,t$
  are the Lagrangian coefficients which may be efficiently computed given
  $(id_{i_1},\cdots,id_{i_t})$, it computes
  $y_1^{s_i}=\prod_{v=1}^t\left(y_1^{f_i(id_{i_v})}\right)^{L_v},i=1,\cdots,l$
  and outputs $\mathbf{x}=(x_1,\cdots,x_l)$, where $x_i=1$, if $y_{i+1}/y_1^{s_i}=g$, $i=1,\cdots,l$, and $x_i=0$, if $y_{i+1}/y_1^{s_i}=1$, $i=1,\cdots,l$.
\end{itemize}

\begin{lemma}
The algorithms give a collection of $(n,t,l,l-\lg p)$-TLTDFs under the DDH assumption.
\end{lemma}

\begin{proof}
The $(n,t)$-threshold secret sharing scheme holds the perfect privacy. 
Both modes of LTDF are computationally indistinguishable.
Therefore, we can show the indistinguishability between
injective and lossy mode of TLTDF.

We transform the inversion algorithm
into threshold version
which does not change the lossy mode.
In the lossy mode, 
the number of possible function outputs is at most $p$, 
the residual leakage $r\leq \lg p$, 
and the lossiness is $k=n-r\geq l-\lg p$.
\end{proof}

\subsection{Instantiation of TLTDF Based on the LWE Assumption}\label{TLTDF Based on LWE}

We recall a variant of LWE-based symmetric key cryptosystem~\cite{PW08} which has a small message space.
Let $\mathbb{T}=\mathbb{R}/\mathbb{Z}$, $\eta\in\mathbb{N}$.
For every message $m\in \mathbb{Z}_p$, we define the ``offset'' $c_m=m/p\in \mathbb{T}$.
The secret key
is $\mathbf{z}\leftarrow \mathbb{Z}^d_q$.
To encrypt $m\in \mathbb{Z}_p$, we choose $\mathbf{a}\leftarrow \mathbb{Z}^d_q$ and an error term $e\leftarrow \chi$.
The ciphertext is
$$
E_\mathbf{z}(m, u; \mathbf{a}, e)
=(\mathbf{a}, \langle \mathbf{a}, \mathbf{z}\rangle + qc_m+u + e) \in \mathbb{Z}^d_q\times \mathbb{Z}_q
$$
where the rounding error $u=\lfloor qc_m\rceil-qc_m \in [-1/2, 1/2]$.
For a ciphertext $c=(\mathbf{a},c')$, the decryption algorithm computes
$t=\eta(c'-\langle \mathbf{a},\mathbf{z}\rangle)/q$ and
outputs $m\in \mathbb{Z}_p$, such that $t-\eta c_m$
is closest to $0$. Note that for any ciphertext,
as long as the absolute total error $|\eta e+\eta u|\leq \eta q/2p$, the decryption is correct.

We use ``matrix encryption'' mechanism in~\cite{PW08} to generate the ciphertext
$$
\mathbf{C}=E_\mathbf{Z}(\mathbf{M},\mathbf{U};\mathbf{A},\mathbf{E})
$$
where
$\mathbf{M}=(m_{i,j})\in \mathbb{Z}_p^{h\times w}$ is a message matrix,
$\mathbf{U}=(u_{i,j})$ is a matrix of rounding errors,
$\mathbf{E}=(e_{i,j})\in \mathbb{Z}_q^{h\times w}$ is error matrix, $e_{i,j}\leftarrow \chi$,
choose independent
$\mathbf{z}_j\leftarrow \mathbb{Z}_q^d$, $\mathbf{Z}=(\mathbf{z}_1,\cdots,\mathbf{z}_w)$,
for each row $i\in[h]$ of the random matrix $\mathbf{A}\in \mathbb{Z}_q^{h\times d}$,
choose independent $\mathbf{a}_i\leftarrow \mathbb{Z}_q^d$.

In the injective mode, the message matrix $\mathbf{M}$ is a matrix $\mathbf{B}$,
which is the tensor product
$\mathbf{I}\otimes \mathbf{b}$, where $\mathbf{I}\in \mathbb{Z}_p^{w\times w}$ is the identity and
$\mathbf{b}=(1,2\cdots,2^{l-1})^T\in \mathbb{Z}_p^l$, $l=\lfloor\log p\rfloor$, $w=h/l$.
In the lossy mode, the message matrix $\mathbf{M}$ is all-zeros matrix $\mathbf{0}$.

\begin{lemma}(\cite{PW08}, Lemma 6.2).
For $h,w=\mathsf{poly}(d)$, the matrix encryption scheme
produces indistinguishable ciphertexts under the assumption that LWE$_{q,\chi}$ is hard.
\end{lemma}

\begin{trivlist}
\item\textbf{Construction.}
We describe a LWE-based TLTDF as follows. 
By using the 
technique of “clearing out the denominator” to bound the quantity of errors, 
we require that the identity space $ID=[n]$, $n\in \mathbb{N}$ and set $\eta=(n!)^3$.
\end{trivlist}

\begin{itemize}
  \item $\mathsf{Samp}_{\textnormal{inj}}$:
  On input $1^d$, it generates
$\mathbf{C}=E_\mathbf{Z}(\mathbf{B},\mathbf{U};\mathbf{A},\mathbf{E})$  
  and outputs the function index $\mathbf{C}$
  and the master trapdoor $mtd=(\mathbf{z}_i,\mathbf{D}_i)$,
  where $\mathbf{z}_i=(z_j^{(i)})$,  
  $\mathbf{D}_i=(b^{(i)}_{jk})$, $z_j^{(i)}, b^{(i)}_{jk}\leftarrow \mathbb{Z}_q$,  
  $i=1,\cdots,w,j=1,\cdots,d,k=1,\cdots,t-1$.
  \item $\mathsf{Samp}_{\textnormal{loss}}$:
  On input $1^d$, it generates
$\mathbf{C}=E_\mathbf{Z}(\mathbf{0},\mathbf{U};\mathbf{A},\mathbf{E})$  
  and outputs the function index $\mathbf{C}$
  and the master trapdoor $mtd=(\mathbf{z}_i,\mathbf{D}_i)$,
  where $\mathbf{z}_i=(z_j^{(i)})$,  
  $\mathbf{D}_i=(b^{(i)}_{jk})$, $z_j^{(i)}, b^{(i)}_{jk}\leftarrow \mathbb{Z}_q$,  
  $i=1,\cdots,w,j=1,\cdots,d,k=1,\cdots,t-1$.
  \item $\mathsf{Share}$:
  On input the master trapdoor $mtd$
  and any identity $id_{i_v}=i_v\in[n]$,
  it sets
  $f_j^{i}(x)=z_j^{(i)}+b^{(i)}_{j1}x+\cdots+b^{(i)}_{j(t-1)}x^{t-1},j=1,\cdots,d$

  and outputs

\begin{equation*}
\begin{aligned}
td_{i_v}=\left(\begin{array}{cccc}
                  f^{(1)}_1(i_v) & f^{(2)}_1(i_v) & \cdots & f^{(w)}_1(i_v) \\
               \vdots & \vdots & \ddots & \vdots \\
                  f^{(1)}_d(i_v) & f^{(2)}_d(i_v) & \cdots & f^{(w)}_d(i_v)
                \end{array}
  \right)
\end{aligned}
\end{equation*}

  \item $\mathsf{F}$: On input the function index $\mathbf{C}$
  and $\mathbf{x}\in \{0,1\}^h$,
  it outputs the vector $\mathbf{a}$$=\mathbf{x}\mathbf{A}$ and 
  $\mathbf{y}=\mathbf{x}\mathbf{C}$.
  \item $\mathsf{F}^{-1}$: On input any 
  shared trapdoor $td_{i_v}$ and $\mathbf{a}=\mathbf{x}\mathbf{A}$,
  it outputs the inversion share  
  $$\delta_{i_v}^{(i)}=\left\langle \mathbf{a}, \left(\begin{array}{c}
                                           f^{(i)}_1({i_v}) \\
                                           \vdots \\
                                           f^{(i)}_d({i_v})
                                         \end{array}
   \right)\right\rangle+e_{i_v}^{(i)}$$
   where $i=1,\cdots,w$.

  \item $\mathsf{CombineF}^{-1}$: 
  On input $ek$, $\mathbf{x}\in \{0,1\}^h$, 
  any $t-1$ inversion shares $\delta_{i_v},v=1,\cdots,t-1$ 
  and identity $id_{i_t}$.
  Because of
  $f^{(i)}_j(i_t)=\sum_{v=0}^{t-1}L_v f^{(i)}_j(i_v)$, $j=1,\cdots,d, i=1,\cdots,w$, where $L_v$, $v=0,1,\cdots, t-1$
  are Lagrangian coefficients which may be efficiently computed given
  $i_0=0,i_1,\cdots, i_{t-1}$,
  it computes the image $\mathbf{y}=(y_1,\cdots,y_w)$ of $\mathbf{x}$ and 
  $\mathbf{x}\mathbf{B}=(m_1,\cdots,m_w)$. For every $i=1,\cdots,w$, 
  $\delta_{i_0}^{(i)}=\langle \mathbf{a}, \mathbf{z}_i \rangle+e_i=y_i-\lfloor qc_{m_i}\rceil
  =\langle \mathbf{a}, \mathbf{z}_i \rangle+(\mathbf{xE})_i$.
    It outputs the inversion share $\delta_{i_t}=\sum_{v=0}^{t-1}L_v\delta_{i_v}$.

  \item $\mathsf{Combine}$: On input any $t$ inversion shares $\delta_{i_1},\cdots, \delta_{i_t}$.
  Because of
  $f^{(i)}_j(0)=\sum_{v=1}^{t}L_vf^{(i)}_j(i_v)$, $j=1,\cdots,d, i=1,\cdots,w$, where $L_v$, $v=1,\cdots, t$
  are Lagrangian coefficients which
  can be efficiently computed given any $t$ identities $i_1,\cdots, i_t$,
  it computes

\begin{equation*}
\begin{aligned}
L_v\delta_{i_v}^{(i)}=\left\langle \mathbf{a},L_v\left(\begin{array}{c}
                                          f^{(i)}_1(i_v) \\
                                          \vdots \\
                                          f^{(i)}_d(i_v)
                                        \end{array}
  \right)\right\rangle+L_ve_{i_v}^{(i)},
\end{aligned}
\end{equation*}

where $v=1,\cdots, t$.

  \begin{equation*}
  \begin{aligned}
  y_i'=\sum_{v=1}^t L_v\delta_{i_v}^{(i)}&=\left\langle\mathbf{a},\left(\begin{array}{c}
                           L_1f_1^{(i)}(i_1)+\cdots+L_tf_1^{(i)}(i_t) \\
                           \vdots \\
                           L_1f_d^{(i)}(i_1)+\cdots+L_tf_d^{(i)}(i_t)
                         \end{array}
  \right) \right\rangle+\sum_{v=1}^t L_ve_{i_v}^{(i)}\\
&=\left\langle \mathbf{a},\mathbf{z}_i\right\rangle+\sum_{v=1}^t L_ve_{i_v}^{(i)}
  \end{aligned}
  \end{equation*}

where $i=1,\cdots,w$ and gets $\mathbf{y'}=(y_1',\cdots,y_w')$,
then it computes $y_i''=\eta(y_i-y'_i)/q,i=1,\cdots,w$
and obtains $m_i\in \mathbb{Z}_p$ such that $y''-\eta c_{m_i}$ is 
closest to $0$.
Finally, it outputs $\mathbf{x}\in\{0,1\}^h$,
so that $\mathbf{x}\mathbf{B}=(m_1,\cdots,m_w)$.
\end{itemize}
We show correctness and lossy properties of our TLTDF as follows. 

We recall some probability distributions in~\cite{PW08}.
For $\alpha \in \mathbb{R}^+$, let $\Psi_\alpha$ be a normal variable with mean $0$ and
standard deviation $\frac{\alpha}{\sqrt{2\pi}}$ on $\mathbb{T}$.
For any probability $\phi: \mathbb{T}\rightarrow \mathbb{R}^+$ and $q\in \mathbb{Z}^+$,
let its discretization
$\bar{\phi}: \mathbb{Z}_q\rightarrow \mathbb{R}^+$ be the discrete distribution over $\mathbb{Z}_q$ of
the random variable $\lfloor q\cdot X_\phi\rceil \bmod q$, where $X_\phi$ is the distribution $\phi$.

\begin{lemma}
Let $q \geq 4p(h+\gamma)$, $\alpha \leq 1/(16p(h + g))$ for $g\geq \gamma^2$,
$\gamma=\sum^t_{v=1}\eta L_v$, where $L_v$, $v=1,\cdots,t$ is the Lagrangian coefficient.
The error matrix $\mathbf{E} = (e_{i,j})\in \mathbb{Z}_q^{h\times w}$ is generated by choosing independent
error terms $e_{i,j}\leftarrow \chi=\bar{\Psi}_\alpha$
and $e_{i_v}\leftarrow \chi$, $v=1,\cdots,t$.
Every entry of $\mathbf{xE}+\sum^t_{v=1}\eta L_ve_{i_v}$ has absolute value less than $q/4p$
for all $\mathbf{x} \in \{0, 1\}^h$,
except with probability at most $w\cdot2^{-g}$
over the choice of $\mathbf{E}$ and $e_{i_v}$.
\end{lemma}

\begin{proof}
By definition, $e_i=\lfloor qs_i \rceil \bmod q$, 
$e_{i_v}=\lfloor qs_{i_v} \rceil \bmod q$
where $s_i,s_{i_v}$ are independent normal variables with 
mean $0$ and variance $\alpha^2$ for each $i\in [h]$, $v\in[t]$. 
Let $s'=\langle \mathbf{x},\mathbf{e}\rangle +\sum^t_{v=1}\eta L_v e_{i_v}$, 
where $\mathbf{e}=(e_1,\cdots,e_h)^T$. 
Then $s'$ is at most $(h+\gamma)/2\leq q/8p$ away from 
$q((\langle \mathbf{x},\mathbf{s}\rangle+\sum^t_{v=1}\eta L_v s_{i_v})\bmod 1)$.

Since the $s_i,s_{i_v}$ are independent, 
$\langle \mathbf{x},\mathbf{s}\rangle+\sum^t_{v=1}\eta L_v s_{i_v}$ 
is distributed as a normal variable with mean $0$ 
and variance at most $(h+\gamma^2)\alpha^2\leq (h+g)\alpha^2$, 
where $\gamma^2>\sum^t_{v=1}(\eta L_v)^2$, 
hence a standard deviation of at most 
$(\sqrt{h+g})\alpha$. 
Then by the tail inequality on normal variables and the hypothesis on $\alpha$,

  \begin{equation*}
  \begin{aligned}
\Pr[|\langle \mathbf{x},\mathbf{s}\rangle+\sum^t_{v=1}\eta L_v s_{i_v}|\geq 1/8p]&\leq \Pr[|\langle \mathbf{x},\mathbf{s}\rangle+\sum^t_{v=1}\eta L_v s_{i_v}|\geq 2\sqrt{h+g}(\sqrt{h+g})\alpha]\\
&\leq \frac{\textnormal{exp}(-2(h+g))}{2\sqrt{h+g}}<2^{-(h+g)}.
  \end{aligned}
  \end{equation*}

We show that for any fixed $\mathbf{x}\in \{0,1\}^h$, 
$\Pr[|s'|\geq q/4p]\leq 2^{-(h+g)}$.
Taking a union bound over all $\mathbf{x}\in \{0,1\}^h$, 
we can conclude that $|s'|<q/4p$ for all $\mathbf{x}\in \{0,1\}^h$ except 
with probability at most $2^{-g}$.

Therefore, for each column $\mathbf{e}$ of $\mathbf{E}$ and $e_{i_v},v=1,\cdots,t$,  
$|s'|<q/4p$, for all $\mathbf{x}$ 
except with probability at most $2^{-g}$ over the choice of $\mathbf{e}$ and 
$e_{i_v},v=1,\cdots,t$. The lemma follows by a union bound over all $w$ columns of $\mathbf{E}$.
\end{proof}

\begin{trivlist}
\item \textbf{Parameters.}
Instantiate the parameters: let $p=h^{c_1}$ for constant $c_1>0$,
$h=d^{c_3}$ for constant $c_3>1$, $\gamma = \sum_{v=1}^t \eta L_v$,
where $L_v$ is the Lagrangian coefficient,
$e_{i_v}\leftarrow \chi, v=1,\cdots,t$,
let $\chi=\bar{\Psi}_\alpha$ where $\alpha\leq 1/(32ph)$
and let $q\in [2\sqrt{d}/\alpha, \mathcal{O}(ph^{c_2})]$ for constant $c_2>1$.

Note that 
for $\mathbf{A}\in \mathbb{Z}_q^{h\times d}$,
the size of the function index is $hd\log q=d^{c_3+1}\log q=\Omega(d^2\log d)$
and for $(\mathbf{x}\mathbf{A},\mathbf{x}\mathbf{C})\in \mathbb{Z}_q^d\times \mathbb{Z}_q^w$,
the size of the image is
$(h+w)\log q=(d^{c_3}+d^{c_3}/\lfloor\log p\rfloor)\log q
=\Omega(d\log d)$.
\end{trivlist}

\begin{trivlist}
\item\textbf{Correctness.}
We now show correctness of the above TLTDF by proving the following theorem.
\end{trivlist}

\begin{theorem}
The TLTDF with above parameters instantiated satisfies the correctness.
\end{theorem}
\begin{proof}
The combining algorithm computes $\mathbf{y'}=(y_1',\cdots,y_w')$ as follows:
$$
y_i'=\left\langle \mathbf{a},\mathbf{z}_i\right\rangle+\sum_{v=1}^t L_ve_{i_v}^{(i)}
=\left\langle \mathbf{a},\mathbf{z}_i\right\rangle+\sum_{v=1}^t L_ve_{i_v}^{(i)}.
$$

We have
\begin{displaymath}
  \begin{split}
y_i''
=\frac{|\eta (y_i-y_i')|}{q}=\frac{|\eta c_{m_i}q+\eta(\mathbf{x}\mathbf{U})_i+\eta(\mathbf{x}\mathbf{E})_i
-\sum_{v=1}^t \eta L_ve_{i_v}^{(i)}|}{q}.
  \end{split}
  \end{displaymath}
Let $g=h\geq \gamma^2$ in above Lemma 7, the absolute total error
\begin{displaymath}
  \begin{split}
|(\mathbf{x}\mathbf{U})_i+(\mathbf{x}\mathbf{E})_i
-\sum_{v=1}^t \eta L_ve_{i_v}^{(i)}|&\leq |(\mathbf{x}\mathbf{U})_i|+(|(\mathbf{x}\mathbf{E})_i|+
|\sum_{v=1}^t \eta L_ve_{i_v}^{(i)}|)\\
&\leq \frac{q}{8p}+\frac{q}{4p}<\frac{q}{2p}.
  \end{split}
  \end{displaymath}
We have
\begin{displaymath}
\begin{split}
|\eta(\mathbf{x}\mathbf{U})_i+\eta(\mathbf{x}\mathbf{E})_i
-\sum_{v=1}^t \eta L_ve_{i_v}^{(i)}|&\leq \eta|(\mathbf{x}\mathbf{U})_i|+(\eta|(\mathbf{x}\mathbf{E})_i|+
\eta|\sum_{v=1}^t \eta L_ve_{i_v}^{(i)}|)\\
&<\frac{\eta q}{2p}.
  \end{split}
  \end{displaymath}
Therefore, the inversion is correct.
\end{proof}

\begin{theorem}
The TLTDF with
above parameters produces indistinguishable function indexes
under the $LWE_{q,\chi}$ assumption.
Moreover, the algorithms give a collection of $(n,t,h,k)$-TLTDFs
under the $LWE_{q,\chi}$ assumption is hard. The residual leakage $r=h-k$ is
$$
r\leq \left(\frac{c_2}{c_1}+o(1)\right)\cdot h.
$$

\end{theorem}
\begin{proof}
The $(n,t)$-threshold secret sharing scheme holds the perfect privacy 
and the injective and lossy modes of LTDF are indistinguishable~\cite{PW08}.
Therefore, we can show the indistinguishability between
injective and lossy mode of TLTDF.

We transform the inversion algorithm into
threshold version which does not change the lossy mode.
In the lossy mode, as in the correctness argument, 
$|(\mathbf{x}\mathbf{U})_i|+(|(\mathbf{x}\mathbf{E})_i|+
|\sum_{v=1}^t \eta L_ve_{i_v}^{(i)}|)<\frac{q}{2p}$. 
Therefore, for $i\in[w]$, the function output 
$y_i=\langle \mathbf{xA},\mathbf{z}_i \rangle+|(\mathbf{x}\mathbf{U})_i|+0+(|(\mathbf{x}\mathbf{E})_i|+
|\sum_{v=1}^t \eta L_ve_{i_v}^{(i)}|)$ can take at most $q/p$ possible values.
Then the number of possible function outputs is at most $q^d(q/p)^w$.
The proof follows (\cite{PW08}, Theorem 6.4), we omit the details.
\end{proof}

\section{Threshold Trapdoor Relation}\label{Threshold Trapdoor Relation}

We show a relaxation of TTDF called TTDR and prove that TTDR maintains same
applications of constructing TPKE and RPKE.
\begin{definition}
A collection of $(n,t)$-TTDRs is a tuple of polynomial-time algorithms
as follows:
\begin{itemize}
  \item $\mathsf{Gen}(1^\lambda)\rightarrow (ek,mtd)$: 
  The generation algorithm is a probabilistic algorithm that 
  on input the security parameter $1^\lambda$,
  outputs a function index $ek$ and
  a master trapdoor $mtd$.

  \item $\mathsf{Share}(mtd,id_i)\rightarrow td_i$:
  The sharing algorithm is a deterministic algorithm that 
  on input the master trapdoor $mtd$ and any identity $id_i,i\in[n]$,
  outputs the shared trapdoor $td_i,i\in[n]$.

  \item $\mathsf{Samp}(ek)\rightarrow (x,y)$:
  On input the function index $ek$,
  the relation sampling algorithm samples a relation $(x,y=\mathsf{F}(ek,x))$.

  \item $\mathsf{F}^{-1}(td_i,y)\rightarrow \delta_{i}$:
  On input any shared trapdoor $td_{i}$
  and an image $y$,
  the partial inversion algorithm outputs the inversion share $\delta_{i}$.

  \item $\mathsf{CombineF}^{-1}(x,y,\delta_{i_1},\cdots,\delta_{i_{t-1}},id_{i_t})\rightarrow \delta_{i_t}$:
  On input $x\in \{0,1\}^l$, its image $y$,  
  any $t-1$ inversion shares $\delta_{i_1},\cdots,\delta_{i_{t-1}}$, and identity $id_{i_t}$, 
  the combining inversion algorithm outputs the inversion share $\delta_{i_t}$ of identity $id_{i_t}$.

  \item $\mathsf{Combine}(\delta_{i_1},\cdots,\delta_{i_{t}},y)\rightarrow x$:
  On input any $t$ inversion shares $\delta_{i_j},j=1,\cdots,t$ and the image $y$,
  the combining algorithm outputs $x$.
\end{itemize}
\end{definition}
we require that in the partial inversion algorithm, 
the combining inversion algorithm and the combining algorithm,
the behavior of the algorithms is unspecified,
if the $y$ is not in the image.

\begin{trivlist}
\item \textbf{Correctness.} For any $id_i, i\in[n]$, $(ek, mtd)\leftarrow \mathsf{Gen}(1^\lambda)$, $td_i\leftarrow \mathsf{Share}(mtd,id_i)$,
any relation $(x,y=\mathsf{F}(ek,x))$,
we require that 
for any $t$ shared trapdoors $td_{i_1}$, $\cdots$, $td_{i_t}$, we have 
$$
x=\mathsf{Combine}(\mathsf{F}^{-1}(td_{i_1},y),\cdots,\mathsf{F}^{-1}(td_{i_t},y),y).
$$

\end{trivlist}

\begin{trivlist}
\item \textbf{Security.}
Let $\mathcal{A}$ be a PPT adversary 
and define its advantage function $\textbf{Adv}^{\text{tow}}_{\text{TTDR},\mathcal{A}}(\lambda)$ as
\end{trivlist}

\begin{displaymath}
  \begin{split}
\textbf{Adv}^{\text{tow}}_{\text{TTDR},\mathcal{A}}(\lambda)=\Pr\left[x=x':\begin{array}{l}
                  (ek,mtd)\leftarrow \mathsf{Gen}(1^\lambda);  \\
                  (x,y)\leftarrow \mathsf{Samp}(ek); \\
                   x'\leftarrow \mathcal{A}^{\textsf{Share}(mtd, \cdot)}(ek,y)\\
                \end{array}
\right]
  \end{split}
  \end{displaymath}
Here, $\mathsf{Share}(mtd, \cdot)$ denotes a shared trapdoor oracle that given an input of any identity $id$, and outputs a shared trapdoor $td_{id}$. The adversary can query the oracle at most $t-1$ times adaptively.
A $(n,t)$-TTDR is threshold one-way if
for any PPT adversary the advantage function is negligible.

Following the constructions of TPKE and RPKE from TTDF,
we can
show generic constructions of TPKE and RPKE from TTDR
by running the relation sampling algorithm of TTDR instead of the evaluation algorithm of TTDF in the encryption algorithm.
The threshold one-wayness ensures that both of TPKE and RPKE are IND-CPA secure.

\begin{trivlist}
\item \textbf{Threshold Lossy Trapdoor Relation.}
Following the definitions of TTDR and TLTDF,
by relaxing the evaluation algorithm of TLTDF into relation sampling algorithm, we present the
definition of TLTDR and show that TLTDR also produces
indistinguishable function indexes.
Similarly, we can prove TLTDR implies TTDR.

We propose a refined definition of the relation by omitting the public computable injective map
in LTDR~\cite{XLLL14}.
Informally, the function index $ek$ is a composite function description which consists of the inverse map
of the public computable injective map.
The relation sampling algorithm
outputs a relation $(x,y=\mathsf{F}(ek,x))$. The inversion algorithm
takes in the trapdoor and the image $y=\mathsf{F}(ek,x)$, outputs $x$.
\end{trivlist}

\begin{trivlist}
\item \textbf{Instantiations of TLTDR.}
Following the instantiation of TLTDF under the DDH assumption and the instantiation of LTDR~\cite{XLLL14},
we give an efficient instantiation under the DDH
assumption by relaxing evaluation algorithm into relation sampling algorithm.
We constructs the TLTDR by using $2\times 3$ matrix encryption.
The function indexes in the injective mode and in the lossy mode of TLTDR are
\begin{displaymath}
  \begin{split}
 \mathbf{C}_0&=\left(\begin{array}{c@{\ \ }c@{\ \ }c}
         g^{r_1} & g^{r_1s_1} & g^{r_1s_2} \\
         g^{r_2} & g^{r_2s_1} & g^{r_2s_2}
       \end{array}\right),
\mathbf{C}_1=\left(\begin{array}{c@{\ \ }c@{\ \ }c}
         g^{r_1} & g^{r_1s_1}g & g^{r_1s_2} \\
         g^{r_2} & g^{r_2s_1} & g^{r_2s_2}g
       \end{array}\right)
  \end{split}
  \end{displaymath}

For $\mathbf{C}=(c_{ij})_{2\times 3}$ and $(x_1,x_2)\leftarrow \mathbb{Z}_p^2$,
the relation sampling algorithm outputs a relation
$(\mathbf{x}=(g^{x_1}$, $g^{x_2})$, $\mathsf{F}(ek,\mathbf{x})=$$(c_{11}^{x_1}c_{21}^{x_2}$,
$c_{12}^{x_1}c_{22}^{x_2}$, $c_{13}^{x_1}c_{23}^{x_2}))$.
The combining algorithm
computes $\mathbf{x}$ by taking as input any $t$ inversion shares and the image $\mathsf{F}(ek,\mathbf{x})$.
It is not hard to show a collection of
$(n,t,2\log p,\log p)$-TLTDRs under the DDH assumption.
\end{trivlist}

\begin{table*}
\begin{center}
\caption{Comparisons Among TPKE Schemes}\label{table1}
\setlength{\tabcolsep}{0.4mm}{
\begin{tabular}{ccccccccc}
\hline
\multirow{2}{*}{Scheme} &\multirow{2}{*}{pk size}
&\multirow{2}{*}{ciphertext size} &\multirow{2}{*}{Enc} &\multirow{2}{*}{Dec} &\multirow{2}{*}{assumption}&adaptive &generic&\multirow{2}{*}{IND-CCA}\\
&& & & &&corruption&construction &\\
\hline

BD10 &$d^5$ &$d^2$&1Mvp &1SS+1Mvp &LWE&$\times$ &$\times$ &$\times$ \\
XXZ11 &$2nd^3\log d$ &$2nd^2\log d$&1SS+1OTS+$n$TBE&$1\text{Inv}_\text{LTDF}$ &LWE&$\times$&$\surd$&$\surd$\\
BKP13 &$d^2\log d$  &$d\log d$  &$1\text{OTS}+2\text{Mvp}$ &$1\text{Inv}_\text{PSF}$  &LWE&$\surd$&$\times$ &$\surd$ \\
Ours$_1$ & $d^2\log d$ & $d\log d$ &$1\text{Mvp}$ &$1\text{Mvp}$ &LWE&$\surd$ &$\surd$& $\times$\\
Ours$_2$ &$l^2|\mathbb{G}| $ & $l|\mathbb{G}| $&$l^2\text{Exp}$ & $1\text{Exp}$ &DDH& $\surd$&$\surd$ & $\times$\\
Ours$_3$ &$6|\mathbb{G}| $ &$3|\mathbb{G}| $&$6\text{Exp}$ &$1\text{Exp}$ &DDH&$\surd$&$\surd$ &$\times$\\
\hline
\end{tabular}}
\end{center}
$^\ddagger$ $n$ and $d$ denotes the number of users and the dimension of lattice
respectively. Mvp, SS, OTS, TBE, $\text{Inv}_\text{LTDF}$,
$\text{Inv}_\text{PSF}$ and $\text{Exp}$
denote the cost of a matrix-vector product,
a secret sharing, a one-time signature, a tag-based encryption, inverting an image of LTDF,
sampling a preimage of preimage sampleable function and a modular exponentiation respectively.
We construct the scheme of Ours$_1$ and Ours$_2$ from TTDF where $l$ denotes input length of function
and the scheme of Ours$_3$ from TTDR.
\end{table*}

\begin{table*}
\begin{center}
\caption{Comparisons Among RPKE Schemes}\label{table2}
\setlength{\tabcolsep}{2.0mm}{
\begin{tabular}{ccccccc}
\hline
\multirow{2}{*}{Scheme} &\multirow{2}{*}{pk size}&\multirow{2}{*}{ciphertext size}& \multirow{2}{*}{assumption}
&adaptive &generic&\multirow{2}{*}{IND-CCA} \\
& & &&corruption&construction &\\
\hline

NP00 &$|\mathbb{G}|$ &$t|\mathbb{G}|$&DDH & $\surd$& $\times$ & $\times$ \\
DF03 &$(t+2)|\mathbb{G}|$&$2(t+1)|\mathbb{G}|$&DDH & $\surd$ & $\times$& $\surd$ \\
Wee11 &$t|\mathbb{G}|$&$(t+2)|\mathbb{G}|$&DDH & $\times$& $\surd$ & $\surd$\\
Wee11 &$t|\mathbb{Z}^*_N|$&$(t+2)|\mathbb{Z}^*_N|$
&factoring & $\times$ & $\surd$& $\surd$\\
Ours$_3$ &$6|\mathbb{G}|$&$(2t+3)|\mathbb{G}|$& DDH & $\surd$ & $\surd$& $\times$ \\
Ours$_2$ &$l^2|\mathbb{G}|$& $(2t+l)|\mathbb{G}|$&DDH & $\surd$ & $\surd$& $\times$ \\
Ours$_1$ &$d^2\log d$&$d\log d$&LWE & $\surd$ & $\surd$& $\times$ \\
\hline
\end{tabular}}
\end{center}
$^\ddagger$ $t$ and $d$ denote the threshold value
and the dimension of lattice, respectively.
We construct the scheme of Ours$_1$ and Ours$_2$ from TTDF where $l$ denotes input length of function
and the scheme of Ours$_3$ from TTDR.
\end{table*}

\section{Performance Evaluations}
\subsection{Theoretical Analysis}
\begin{trivlist}
\item\textbf{Communication Cost and Efficiency Comparison.}
Table \ref{table1} compares the
communication costs and computational costs
of our lattice-based TPKE schemes with that in~\cite{BD10},~\cite{XXZ11},~\cite{BKP13}.
For lattice-based TPKE, the communication cost of our scheme is less than~\cite{BD10},
in which they need to use a large modulus which causes larger ciphertexts. Compared with~\cite{XXZ11},
they split the message into many pieces and encrypt every piece by a different tag-based encryption,
that cause the size of the public key and the ciphertext is at least linear in the number of users,
while our scheme splits the master secret key directly
and shows the size of the public key and the ciphertext is independent of the number of users.
What's more, the computational
cost of our TPKE is also less than~\cite{BD10},~\cite{XXZ11},~\cite{BKP13},
especially during the encryption and decryption phase,
our TPKE scheme
only requires to compute a simple matrix-vector product respectively.
However, in~\cite{BD10}, the decryption algorithm requires every user computes a sharing by a pseudorandom secret sharing and
a matrix-vector product.
In~\cite{XXZ11}, the encryption algorithm needs to run a secret sharing scheme
to split a message into $n$ pieces, $n$ times tag-based encryption to encrypt every piece and a one-time signature.
Moreover, the decryption algorithm require to check the signature and invert an image of lossy
trapdoor function to obtain a decryption share.
Compared with~\cite{BKP13}, their encryption algorithm requires every user runs a one-time signature and compute twice
matrix-vector product, and the decryption algorithm needs to 
run the inversion algorithm of preimage sampleable function~\cite{GPV08}.
\end{trivlist}

Table \ref{table2} compares the communication costs of our RPKE schemes with that 
in~\cite{NP00},~\cite{DF03},~\cite{Wee11}.
The size of the public key of our DDH-based RPKE
is a $2\times 3$ matrix which is less than
~\cite{DF03} and~\cite{Wee11} in which the size of the public key is at least linear with the 
revocation threshold value. Besides, our RPKE scheme is based on the LWE assumption, 
while the existing RPKE schemes~\cite{NP00},~\cite{DF03},~\cite{Wee11} are 
based on the number-theoretic assumptions.

\subsection{Experimental Analysis}
In order to evaluate the practical performance of our schemes, we implement
the TTDF in Section \ref{Threshold Trapdoor Function}, TPKE in Section \ref{Threshold Encryption from TTDF}
and RPKE in Section \ref{Revocation Encryption from TTDF} 
based on the NTL library.
The program is executed on an Intel Core i7-2600 CPU 3.4GHz and 4GB RAM running Linux Deepin 15.4.1 64-bit system.

\begin{trivlist}
\item \textbf{Experiment Setting and Computation Time.}
As depicted in Table \ref{table3} and Table \ref{table4}, we set
the security parameter $ \lambda=128,256,512$ respectively, and
the dimension of lattice
$d=512,768,1024$ respectively
and compute other parameter by the Section \ref{Instantiations of TLTDF},
where $h=d^{c_3}$, $c_3>1$, $p=h^{c_1}$, $c_1>0$, $l=\lfloor \lg p \rfloor$, $w=m=h/l$,
$\alpha \leq 1/(32ph)$ and $q>2\sqrt{d}/\alpha$. What's more, 
we set the number of users\footnote{
We use the technique of ``clearing out the denominator" to preserve correct decryption
and limit the number of extra errors in $(n!)^4$.
In order to run the LWE-based TPKE and RPKE practically, we set the number of users $n=4$.}
is $n=4$ and
the threshold value is $t=3$ in TPKE, and the number of revoked users is $r=2$ in RPKE.

As depicted in Table \ref{table3} and Table \ref{table4}, we show
the average running times of all algorithms in our TPKE and RPKE schemes.
For different security levels, we set the security parameter $\lambda=128,256,512$ respectively.
The average running times of all algorithms in both DDH-based TPKE and RPKE are the level of milliseconds.
Therefore, our schemes are efficient and practical.
Meanwhile, we set the dimension of lattice $d=512, 768, 1024$ respectively.
In LWE-based TPKE scheme, the average running times of the encryption algorithm and the decryption algorithm are
0.076s, 0.167s, 0.293s, and 0.005s, 0.012s, 0.021s.
In LWE-based RPKE scheme, the average running times of the encryption algorithm are
0.092s, 0.209s, 0.382s.
From these outcomes, we note that 
the encryption algorithm and the decryption algorithm of our TPKE schemes 
and the encryption algorithm of our RPKE schemes are efficient.
\end{trivlist}

\begin{table*}
\begin{center}
\caption{Experiment Setting and Computation Time of DDH-Based TPKE and RPKE}\label{table3}
\setlength{\tabcolsep}{2.5mm}{
\begin{tabular}{c|c|c|c|c|c|c|c|c|c|c}
\hline
\multicolumn{4}{c}{Parameter}& \multicolumn{4}{|c|}{TPKE Time (ms)} & \multicolumn{3}{c}{RPKE Time (ms)}\\
\hline
$\lambda$ & $n$ & $t$ & $r$ & KeyGen & Encrypt& Decrypt & Combine &
 KeyGen & Encrypt & Decrypt\\
\hline
128  & 4 & 3 & 2 & 2.184  & 0.263 & 0.145 & 0.309 & 4.978 & 0.390 & 0.467 \\
\hline
256 & 4 & 3 & 2 & 9.499  & 0.517 & 0.312 & 0.566 & 11.59 & 0.839 & 0.770 \\
\hline
512 & 4 & 3 & 2 & 68.08  & 1.311 & 0.797 & 1.318 & 34.53 & 1.890 & 1.786 \\
\hline
\end{tabular}}
\end{center}
$^\ddagger$ $\lambda$, $n$, $t$ and $r$ indicate the security parameter, 
the number of users, the threshold value, and the number of revoked users, respectively.
\end{table*}

\begin{table*}
\begin{center}
\caption{Experiment Setting and Computation Time of LWE-Based TPKE and RPKE}\label{table4}
\setlength{\tabcolsep}{1.0mm}{
\begin{tabular}{c|c|c|c|c|c|c|c|c|c|c|c|c|c}
\hline
\multicolumn{7}{c}{Parameter}& \multicolumn{4}{|c}{TPKE Time (s)} & \multicolumn{3}{|c}{RPKE Time (s)}\\
\hline
$d$ & $h$ &  $p$&  $w$ & $n$ & $t$ & $r$ & KeyGen & Encrypt& Decrypt & Combine &
KeyGen & Encrypt & Decrypt\\
\hline
512 &  2200 &  2063&  200  & 4 & 3 & 2 & 2.178 & 0.076 & 0.005 & 2.273 & 2.299 & 0.092 & 2.146\\
\hline
768 &  3260 &  6029&  280  & 4 & 3 & 2 & 4.451 & 0.167 & 0.012 & 8.949 & 4.215 & 0.209 & 9.032\\
\hline
1024 &  4420 &  9859&  340  & 4 & 3 & 2 & 7.724 & 0.293 & 0.021 & 17.554 & 7.367 & 0.382 & 17.962\\
\hline
\end{tabular}}
\end{center}
$^\ddagger$ $d$, $h$, $p$, $w$, $n$, $t$ and $r$ indicate the dimension of lattice,
the number of rows of matrix $\mathbf{A}$ of the public key, 
the size of the message space $\mathbb{Z}_p$, 
the number of columns of matrix $\mathbf{Z}$ of the master secret key, the number of users,
the threshold value, and the number of revoked users, respectively.
\end{table*}

\begin{trivlist}
\item \textbf{Acknowledgements.} This work is supported by the National Natural Science Foundation of China (Grant No. 61772522), Youth Innovation Promotion Association CAS, and Key Research Program of Frontier Sciences, CAS (Grant No. QYZDB-SSW-SYS035).
\end{trivlist}

\end{document}